\documentclass{article}
 
\usepackage{amssymb}
\usepackage{amsmath}
\usepackage{amsthm}
\usepackage{color}
\usepackage{graphicx}
\usepackage[algo2e,linesnumbered,vlined,lined,ruled]{algorithm2e}
\usepackage{float}
\usepackage{subfigure}
\usepackage{fullpage}

\floatstyle{ruled}
\restylefloat{figure}
\restylefloat{table}

\newtheorem{definition}{Definition}
\newtheorem{lemma}{Lemma}
\newtheorem{theorem}{Theorem}

\begin{document}

\title{A Review on the Tree Edit Distance Problem and Related Path-Decomposition Algorithms
}

\author{
Shihyen Chen\footnote{
Department of Computer Science,
University of Western Ontario,
London, Ontario, Canada, N6A 5B7. \newline
Email: shihyen\_c@yahoo.ca.}
}
\maketitle

\begin{abstract}
An ordered labeled tree is a tree in which the nodes are labeled
and the left-to-right order among siblings is relevant. The edit
distance between two ordered labeled trees is the minimum cost of
changing one tree into the other through a sequence of edit
steps.
In the literature, there are a class of algorithms based on different yet closely related
path-decomposition schemes. 
This article reviews the principles
of these algorithms, and studies the
concepts related to the algorithmic complexities as a consequence of
the decomposition schemes. 
\end{abstract}
%
\section{Introduction} \label{sec:introduction}
An ordered labeled tree is a tree in which the nodes are labeled
and the left-to-right order among siblings is significant. 

The tree edit distance metric was introduced by Tai as a
generalization of the string editing problem~\cite{StringED}.
Given two trees $T_1$ and $T_2$, the tree edit distance between
$T_1$ and $T_2$ is the minimum cost to change one tree into
the other by a
sequence of edit steps. Tai~\cite{Tai} gave an algorithm with a time complexity of
$O(|T_1|^3 \times |T_2|^3)$.
Subsequently, a number of improved algorithms were 
developed~\cite{Zhang:Shasha:TreeEditDistance:SIAMJC1989,klein:98,wchen:2001,DMRW2009,TED:Pawlik-Augsten:2012,TED:Pawlik-Augsten:2014,Chen-Zhang:TED:2014}.
Bille~\cite{Bille:2005} presented a survey on the tree edit distance algorithms.
This article focuses on a class of algorithms that are based on closely related dynamic programming approaches, developed by
Zhang and Shasha~\cite{Zhang:Shasha:TreeEditDistance:SIAMJC1989}, 
Klein~\cite{klein:98}, and
Demaine \textit{et al.}~\cite{DMRW2009}, with time complexities of
$O(|T_1| \times |T_2| \times \prod_{i=1}^{2} \min\{depth(T_i),\#leaves(T_i)\})$,
$O(|T_1|^2 \times |T_2| \times \log |T_2|)$, and
$O(|T_1|^2 \times |T_2| \times (1 + \log \frac{|T_2|}{|T_1|}))$,
respectively.
The essential features common in these algorithms are:
\begin{enumerate}
\item a postorder enumeration of the subproblems, 
\item the recursive partitioning
of trees into disjoint paths, each associated with a separate subtree-subtree distance 
computation. 
\end{enumerate}
The notions related to these paths as a result of the recursive partitioning were formalized by
Dulucq and Touzet~\cite{Dulucq:Touzet:2005}, and referred to as 
``decomposition strategies''.
The algorithm by Demaine \textit{et al.} yields the best worst-case time complexity. They also showed that there exist  
trees for which 
$\Omega(|T_1|^2 \times |T_2| \times (1 + \log \frac{|T_2|}{|T_1|}))$ time is required to compute
the distance no matter what strategy is used.

In this article, we review and study the concepts underlying various
algorithmic approaches based on 
``decomposition strategies'' as well as their impacts on
the time complexity in computing the tree edit distance. 

The article is organized as follows.
Section~\ref{sec:ted} introduces the problem of tree edit distance,
and gives some initial solutions based on naive strategies.
Section~\ref{sec:strategies} presents improved strategies, focusing
on the conceptual aspects related to the time complexities. 
Section~\ref{sec:conclusion} gives concluding remarks.

\section{Preliminaries} \label{sec:ted}

Before we study the tree edit distance problem, it would be
beneficial to recall the solution for string edit distance
because the tree problem is a generalization
of the string problem, and the solution for the tree problem
may be constructed in ways analogous to the string problem. 
The string edit distance $d(S_1, S_2)$ can be solved by Equation~\ref{eqn:string-dist}
where $u$ and $v$ may be both the last elements or the first elements
of $(S_1, S_2)$. The three basic edit steps are substitution,
deletion, and insertion, with respective costs being
$\delta(u, v)$, $\delta(u, \varnothing)$, and $\delta(\varnothing, v)$.

\begin{definition}[String Edit Distance]
The edit distance $d(S_1, S_2)$ between two strings $S_1$ and $S_2$ is the minimum
cost to change $S_1$ to $S_2$ via a sequence of basic edit steps. 
\end{definition}

\begin{equation} \label{eqn:string-dist}
d(S_1, S_2) = 
    \min\left\{
        \begin{array}{l}
          d(S_1 - u, S_2) + \delta(u, \varnothing), \\
          d(S_1, S_2 - v) + \delta(\varnothing, v), \\
          d(S_1 - u, S_2 - v) + \delta(u, v) 
        \end{array}
        \right\} \enspace .
\end{equation}

We now turn to the tree edit distance.
First, we define some basic notations that will be useful
in the rest of the article.
\par
Given a tree $T$, we denote by $r(T)$ its root and $t[i]$ the $i$th node in $T$. 
The subtree rooted at $t[i]$ is denoted by $T[i]$. 
Denote by $F \circ G$ the left-to-right concatenation of $F$ and $G$.
The notation $F - T$ represents the structure resulted from removing
$T$ from $F$.

\begin{definition}[Tree Edit Distance]
The edit distance $d(T_1, T_2)$ between two trees $T_1$ and $T_2$ is the minimum
cost to change $T_1$ to $T_2$ via a sequence of basic edit steps. 
\end{definition}

Analogous to string editing, there are three basic edit operations on a tree: substitution of which
the cost is $\delta(t_1, t_2)$, insertion of which the cost is
$\delta(\varnothing, t_2)$, and
deletion of which the cost is $\delta(t_1, \varnothing)$.
The substitution operation substitutes a tree node with another one.
The insertion operation inserts a node into a tree. 
If the inserted node is made a child of some node in the tree,
the children of this node become the children of the inserted node.
The deletion operation deletes a node from a tree, and the children
of the deleted node become the children of the parent of the deleted
node.
These operations are displayed in Figure~\ref{fig:tree_edit}.

\begin{figure}[htbp]
  \begin{center} \subfigure[substitution]{\label{fig:edit_r}\scalebox{0.7}{\input{tree_edit_r.pstex_t}}} \\ 
  \subfigure[deletion]{\label{fig:edit_d}\scalebox{0.7}{\input{tree_edit_d.pstex_t}}} \\
  \subfigure[insertion]{\label{fig:edit_i}\scalebox{0.7}{\input{tree_edit_i.pstex_t}}}
  \end{center}
  \caption{Basic tree edit operations.}
  \label{fig:tree_edit}
\end{figure}

The set of substitution steps can be represented as a mapping relation satisfying the following conditions:

\begin{enumerate}
\item One-to-one mapping: A node in one tree can be mapped
to at most one node in another tree.
\item Sibling order is preserved: For any two substitution steps
$(t_1[i] \rightarrow t_2[j])$ and $(t_1[i'] \rightarrow t_2[j'])$ in the edit script, $t_1[i]$ is to the left of $t_1[i']$ if and only
if $t_2[j]$ is to the left of $t_2[j']$ (see Figure~\ref{fig:mapping_s}).
\item Ancestor order is preserved: For any two substitution steps
$(t_1[i] \rightarrow t_2[j])$ and $(t_1[i'] \rightarrow t_2[j'])$ in the edit script, $t_1[i]$ is an ancestor of $t_1[i']$ if and only
if $t_2[j]$ is an ancestor of $t_2[j']$ (see Figure~\ref{fig:mapping_a}).
\end{enumerate}

\begin{figure}[htbp]
  \begin{center}
\subfigure[sibling order preserved]{\label{fig:mapping_s}\includegraphics[scale=0.6]{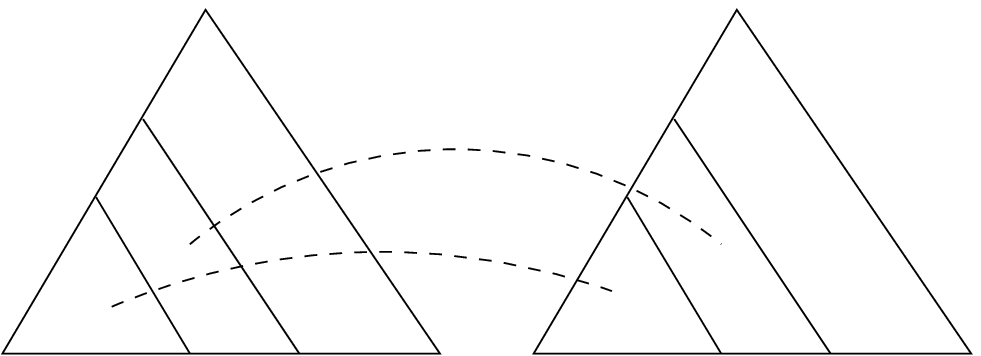}}
\qquad\qquad     
\subfigure[ancestor order preserved]{\label{fig:mapping_a}\includegraphics[scale=0.6]{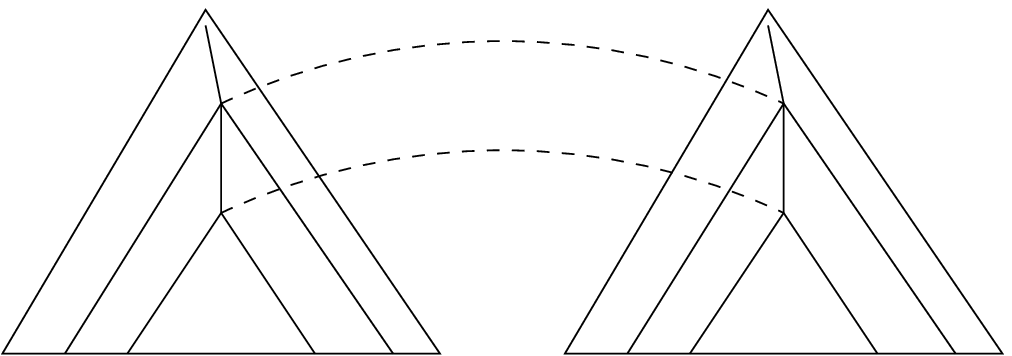}}
  \end{center}
  \caption{Tree editing conditions that preserve sibling orders
  and ancestor orders.}
  \label{fig:mapping}
\end{figure}

As a consequence of these conditions, the substitution steps
are consistent with the structural hierarchy in the original trees.

For the class of algorithms that we consider,
the solution for tree edit distance is based on the recursive formula
for forest edit distance in Equation~\ref{eqn:forest-dist}. 

\begin{equation} \label{eqn:forest-dist}
d(F, G) = 
    \min\left\{
        \begin{array}{l}
          d(F - r(T), G) + \delta(r(T), \varnothing), \\
          d(F, G - r(T')) + \delta(\varnothing, r(T')), \\
          d(F - T, G - T') + d(T, T') 
        \end{array}
        \right\} \enspace .
\end{equation}

A forest as a sequence of subtrees bears resemblance to a string if
each subtree is viewed as a unit of element. 
A string can be represented as a sequence, or an ordered set, of labeled nodes.
A forest reduces to a string when each subtree contains a single node.
In this view, the problem of forest distance may be approached
in ways analogous to the string distance problem, and the
solution would be a generalization of the string solution. 
The meaning of such a solution is based on the principle,
analogous as in the string case, that 
if we know the solutions of some subproblems each of which being a modification from the original problem by one of the three
aforementioned basic operations, then the solution of the original problem
can be constructed from the solutions of these subproblems by means of
a finite number of simple arithmetics. 
The same principle holds
recursively for all the subproblems. 
The tree-to-tree distance $d(T, T')$ in Equation~\ref{eqn:forest-dist}
is computed as in Equation~\ref{eqn:tree-dist}.
Meanwhile, when both forests are composed of one tree (i.e., $(F, G)=(T, T')$), Equation~\ref{eqn:forest-dist} reduces to 
Equation~\ref{eqn:tree-dist} which in turn makes use of
Equation~\ref{eqn:forest-dist} for computing the
associated subforest distances.

\begin{equation} \label{eqn:tree-dist}
d(T, T') = 
    \min\left\{
        \begin{array}{l}
          d(T - r(T), T') + \delta(r(T), \varnothing), \\
          d(T, T' - r(T')) + \delta(\varnothing, r(T')), \\
          d(T - r(T), T' - r(T')) + \delta(r(T), r(T'))
        \end{array}
        \right\} \enspace .
\end{equation}

The recursion in Equation~\ref{eqn:forest-dist} takes on two possible directions (see Figure~\ref{fig:recurse}):

\begin{itemize}
\item leftmost recursion where both $r(T)$ and $r(T')$ are leftmost roots,
\item rightmost recursion where both $r(T)$ and $r(T')$ are rightmost roots.
\end{itemize}

\begin{figure}[htbp]
  \begin{center}
\renewcommand{\thesubfigure}{(a-1)}
\subfigure[leftmost deletion]{\label{fig:recurse_d_l}\includegraphics[scale=0.45]{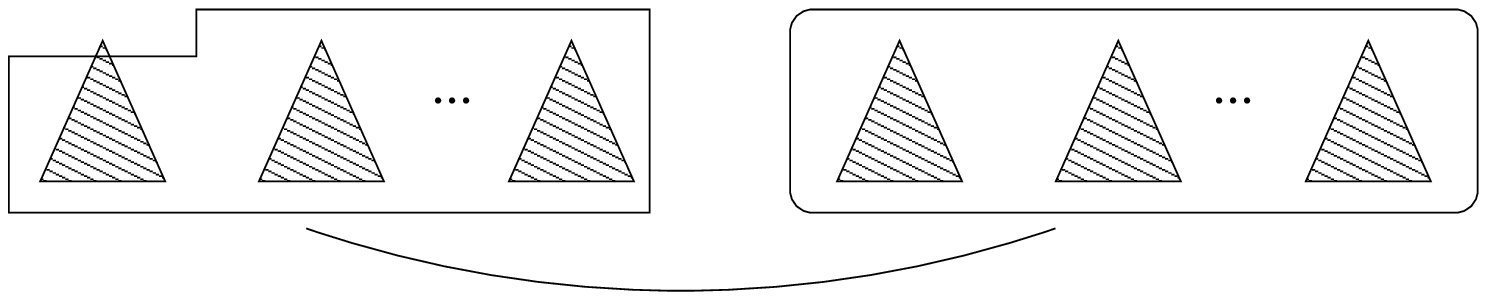}}
\qquad\qquad     
\renewcommand{\thesubfigure}{(a-2)}
\subfigure[rightmost deletion]{\label{fig:recurse_d_r}\includegraphics[scale=0.45]{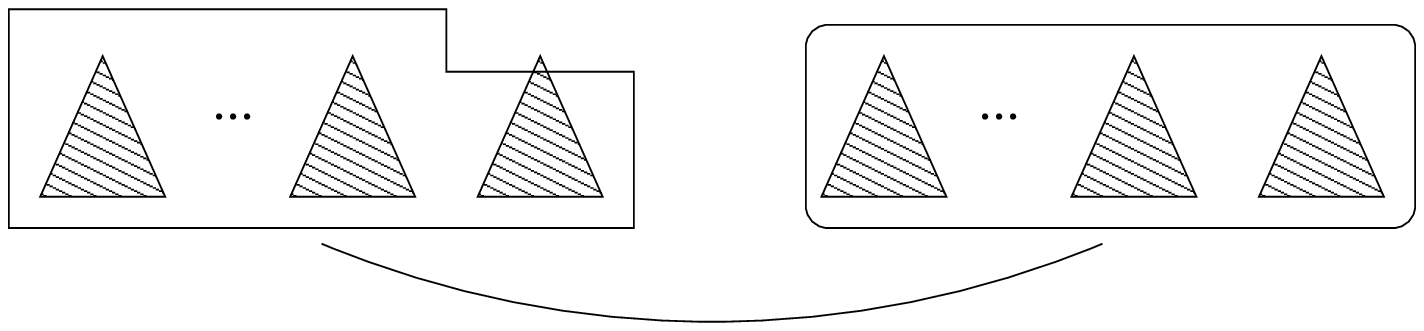}} \\
\renewcommand{\thesubfigure}{(b-1)}
\subfigure[leftmost insertion]{\label{fig:recurse_i_l}\includegraphics[scale=0.45]{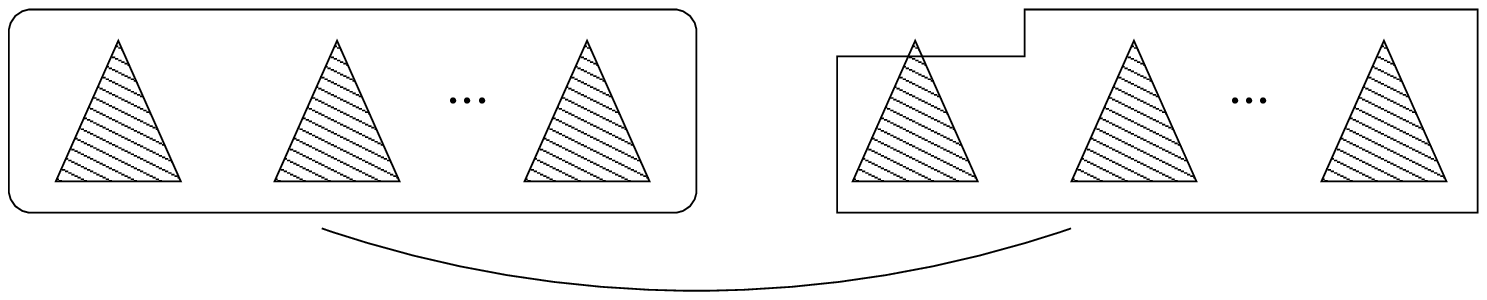}} 
\qquad\qquad     
\renewcommand{\thesubfigure}{(b-2)}
\subfigure[rightmost insertion]{\label{fig:recurse_i_r}\includegraphics[scale=0.45]{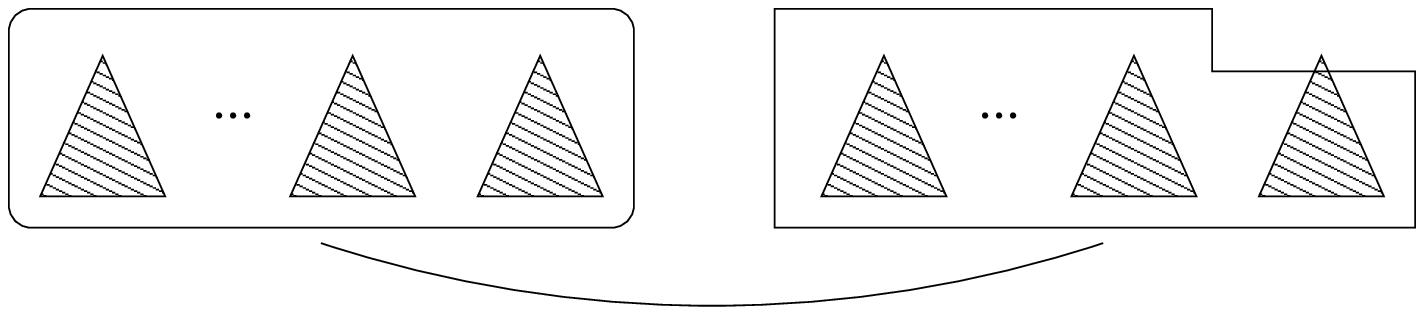}} \\
\renewcommand{\thesubfigure}{(c-1)}
\subfigure[leftmost substitution]{\label{fig:recurse_m_l}\includegraphics[scale=0.45]{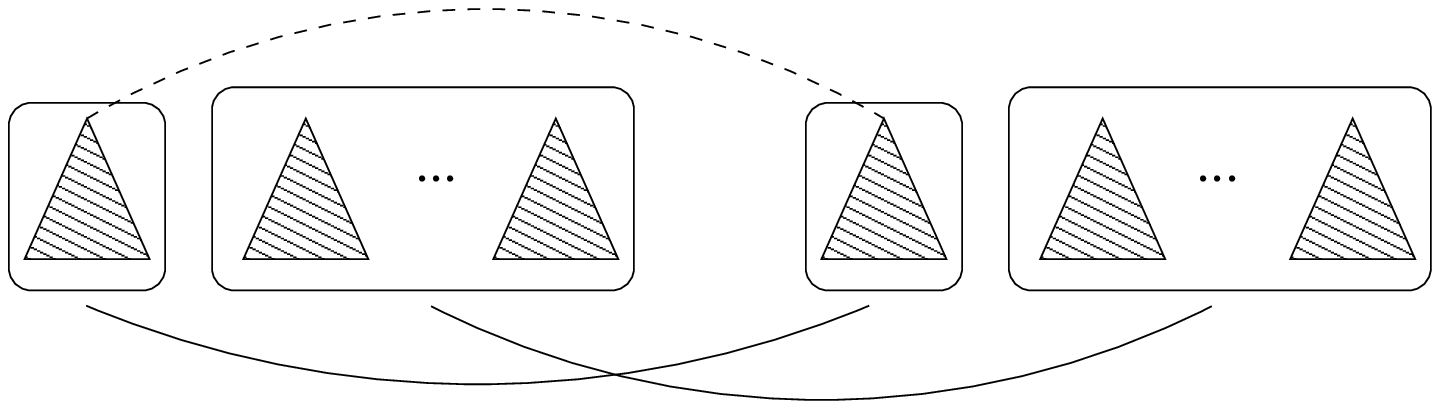}}
\qquad\qquad   
\renewcommand{\thesubfigure}{(c-2)}  
\subfigure[rightmost substitution]{\label{fig:recurse_m_r}\includegraphics[scale=0.45]{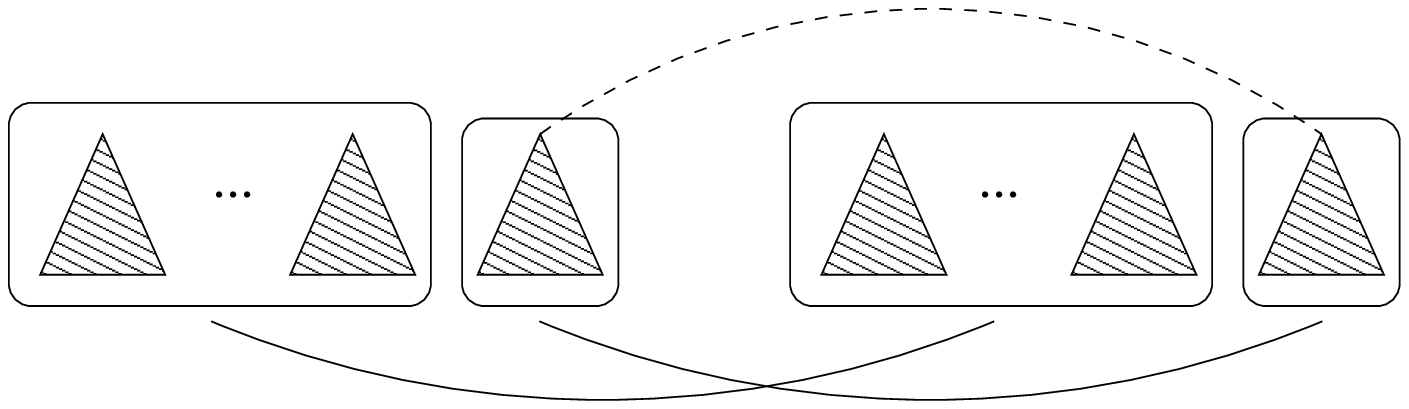}}
  \end{center}
  \caption{Recursion directions.}
  \label{fig:recurse}
\end{figure}

There are a few things to note regarding the above formulae. 
First, we need all the subtree-subtree distances in order to
construct the solution. That is, given
$Q_a=\{T_a[i]\mid t_a[i] \in T_a\}$ with $a \in \{1, 2\}$,
we need to compute all the distances for $Q_1 \times Q_2$.
Since we are solving an optimization problem,
the result is optimal only if all possible cases
have been considered from which the optimal one is selected. 
This means that all combinations of node-to-node mappings
which satisfy the editing conditions need to be considered, which
translates into the need for computing all subtree-subtree distances.
Second, the direction of recursion has an influence on which subforests
would be relevant in the construction of the solution. These are the
subforests that would appear in the recursive calls. 

\begin{definition}[Relevant Subforests] \label{def:relevant-subforests}
The ``relevant subforests'' with respect to a tree edit distance
solution are those subforests that appear in the recursive calls
in Equation~\ref{eqn:forest-dist}.
\end{definition}

Figure~\ref{fig:recurse_rm_d} and Figure\ref{fig:recurse_rm_a}
show examples of relevant subforests generated from rightmost recursion.
Figure~\ref{fig:recurse_heavy_d} and Figure~\ref{fig:recurse_heavy_a}
show examples of relevant subforests generated from a recursion
that operates on the left side and right side intermittently
with respect to a predefined path.
More details will be given in the next section regarding this type
of recursion. 

\begin{figure}[htbp]
\begin{center}
\subfigure[relevant subforests resulted from successive deletions]{\label{fig:recurse_rm_d}\includegraphics[scale=0.5]{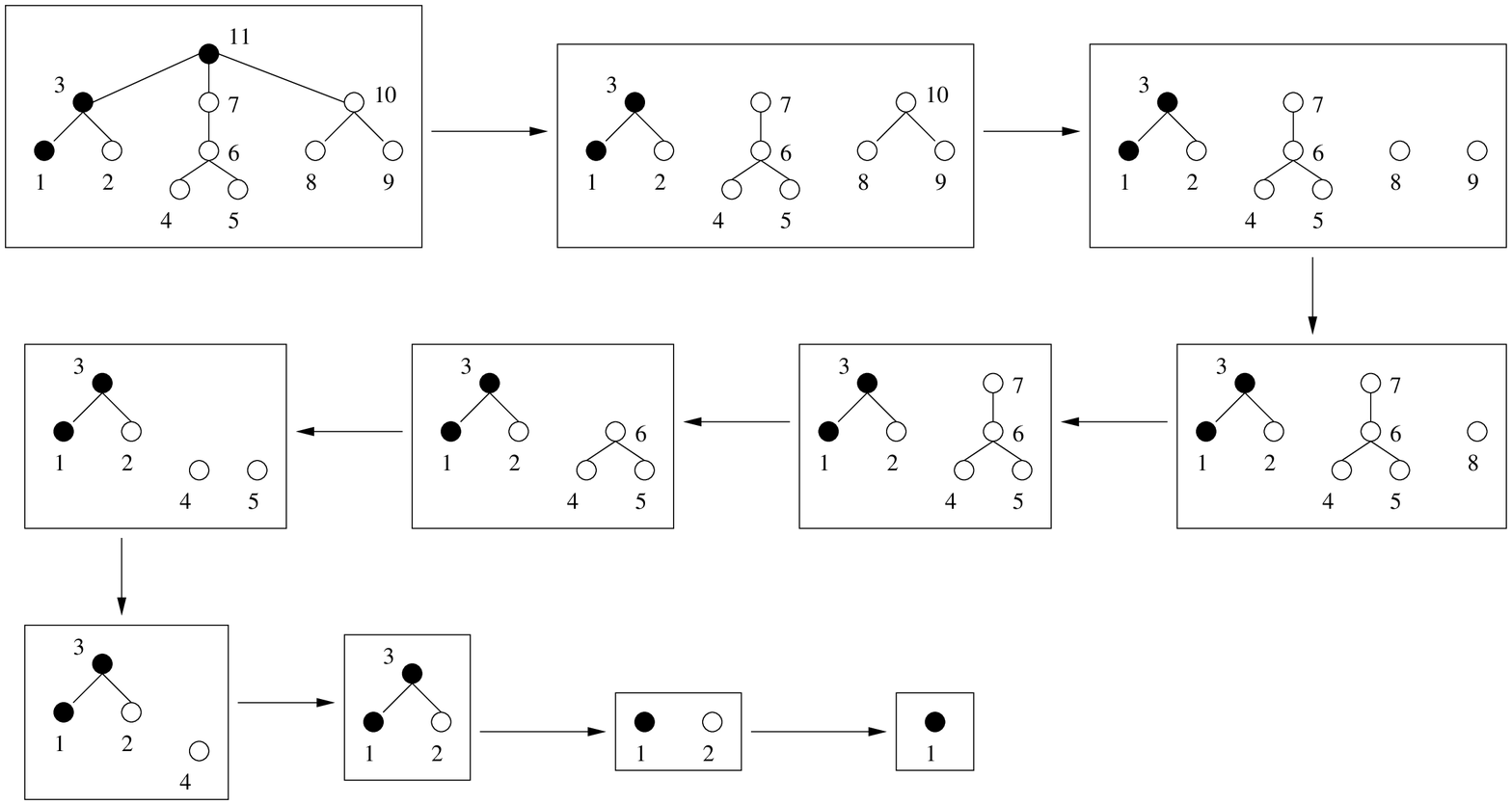}}
\\ 
\subfigure[relevant subforests resulted from deletions and substitutions]{\label{fig:recurse_rm_a}\includegraphics[scale=0.5]{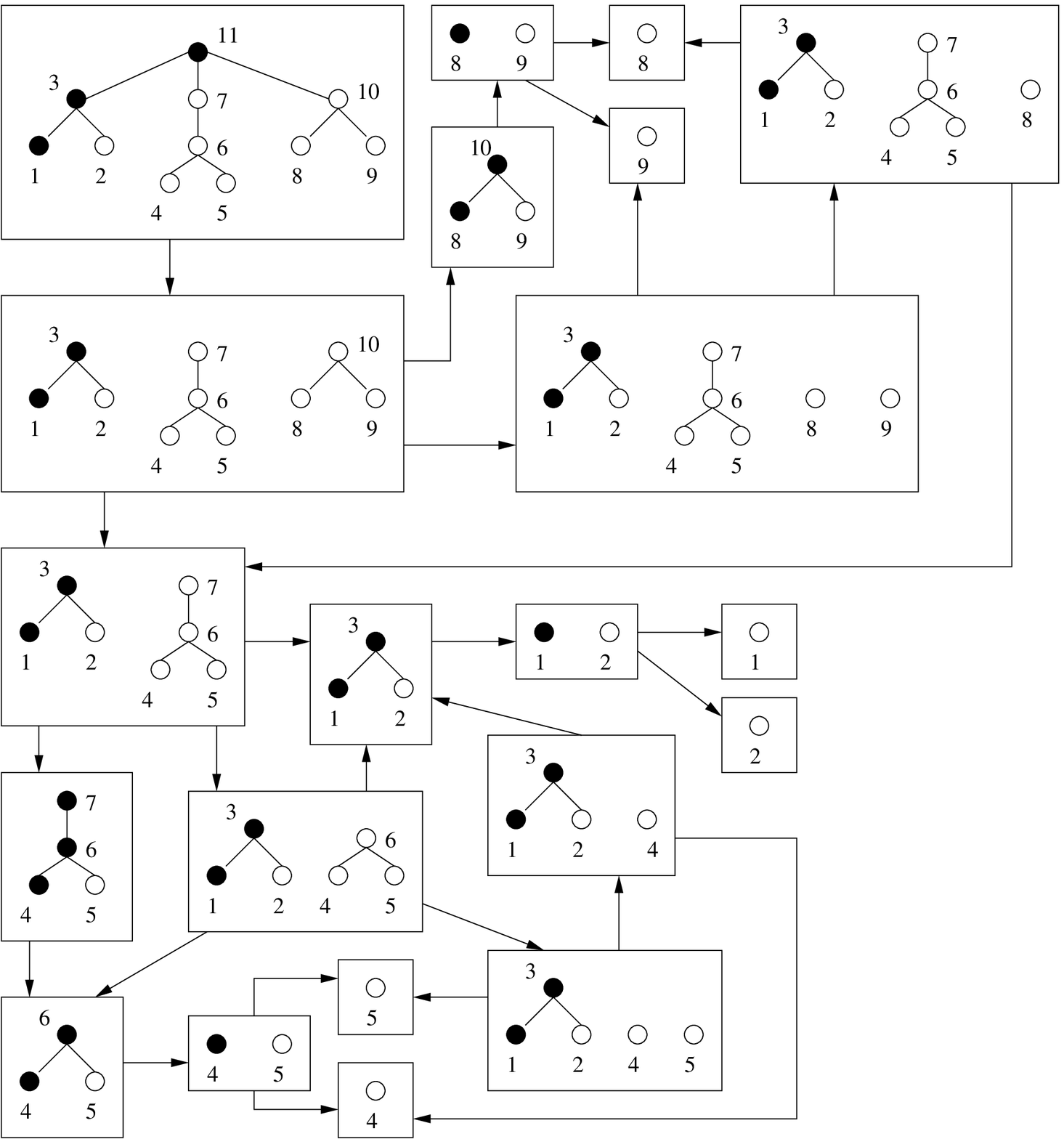}}
\end{center}
\caption{An example showing the relevant subforests from the rightmost recursion
with respect to the leftmost path.}
\label{fig:recurse_rm}
\end{figure}

\begin{figure}[htbp]
\begin{center}
\subfigure[relevant subforests resulted from successive deletions]{\label{fig:recurse_heavy_d}\includegraphics[scale=0.5]{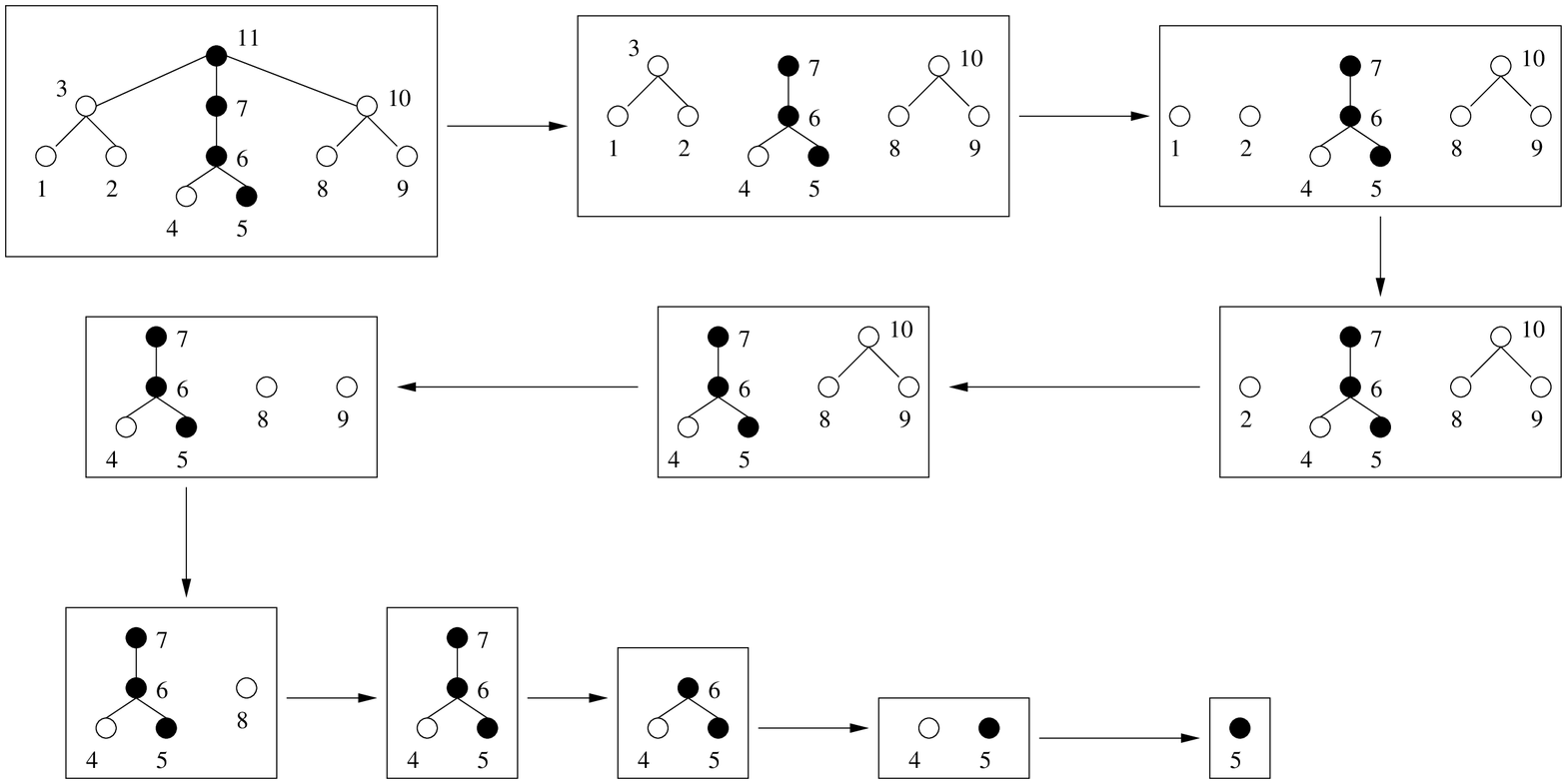}}
\\
\subfigure[relevant subforests resulted from deletions and substitutions]{\label{fig:recurse_heavy_a}\includegraphics[scale=0.5]{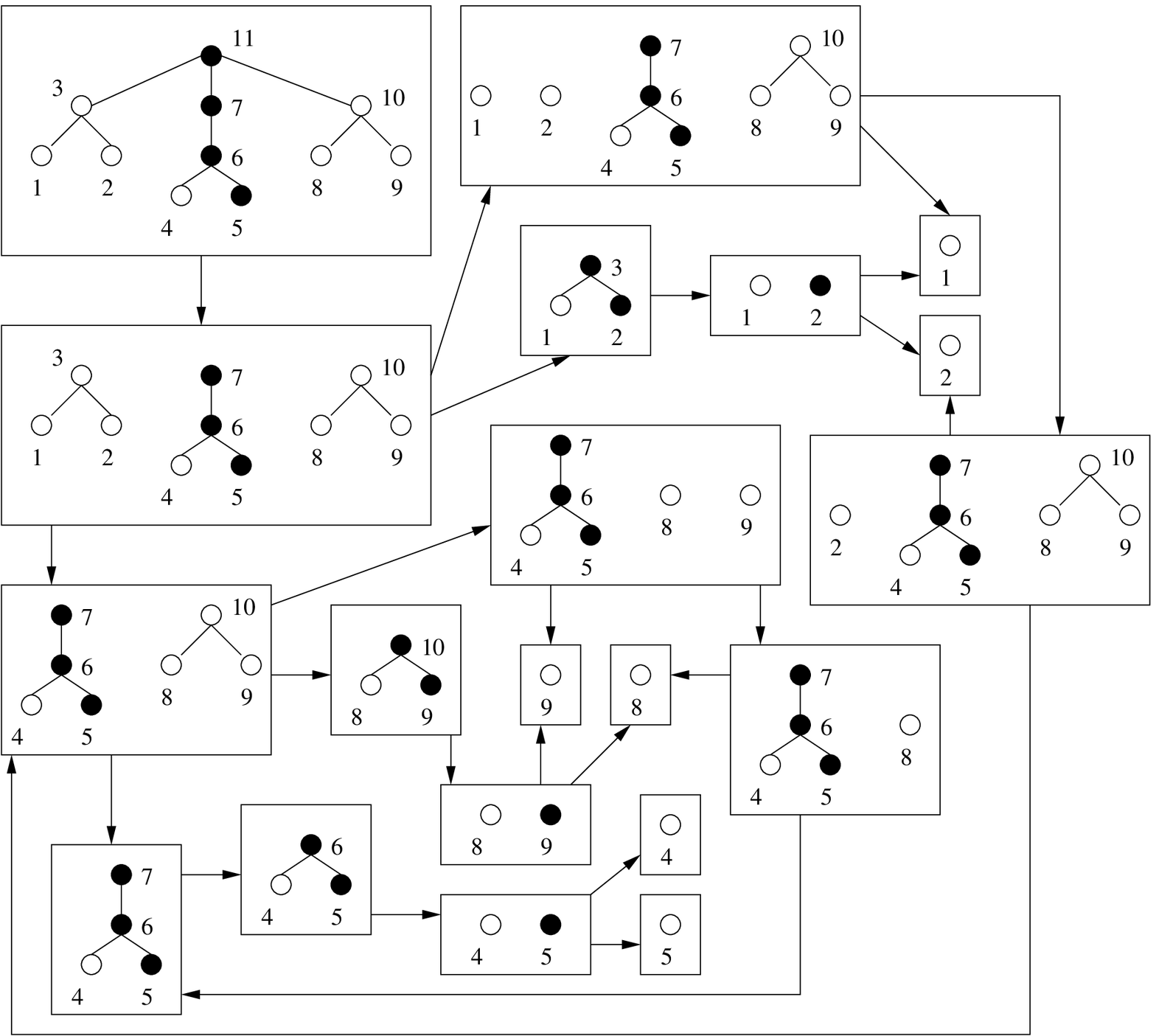}}
\end{center}
\caption{An example showing the relevant subforests from a recursion that operates on the
left side and right side intermittently with respect to a predefined
path.
}
\label{fig:recurse_heavy}
\end{figure}

In constructing an algorithmic solution based on
Equation~\ref{eqn:forest-dist},
there are two complementary aspects to consider:

\begin{itemize}
\item Top-down aspect: This concerns the direction
from the left-hand side to the right-hand side of the
recursion. 
\item Bottom-up aspect: This concerns the direction from the right-hand side
to the left-hand side of the recursion. 
\end{itemize}

In the context of complexity analysis, we express the number of
elementary operations in terms of the number of recursive calls along relevant recursion paths or the number of steps in a bottom-up enumeration sequence,
interchangeably. 
This is due to the fact that to every sequence of
top-down recursive calls based on Equation~\ref{eqn:forest-dist} 
corresponds a sequence of bottom-up enumeration steps.

Our plan in understanding the complexity issues is to start with the bottom-up aspect and eventually
relate it to the top-down aspect. 
As such, we initially consider procedures based on the bottom-up style. 
As a starting point, consider the following approaches:

\begin{itemize}
\item the recursion direction is fixed to be either leftmost
or rightmost, 
\item the recursion direction may vary between leftmost and rightmost.
\end{itemize}

In either approach, we need an enumeration scheme which specifies
the order of distance computations for the subproblems. 

\paragraph{Fixed-Direction Recursion:}

For recursion of fixed direction, a naive scheme
is to arrange the subtree-subtree distance computations,
as well as the relevant forest-forest distance computations,
in one of two alternative ways
as follows:

\begin{itemize}
\item LR-postorder: The subtrees as well as the subforests contained in
each subtree are enumerated in left-to-right postorder.
\item RL-postorder: The subtrees as well as the subforests contained in
each subtree are enumerated in right-to-left postorder.
\end{itemize}

The procedures for sorting the enumeration order for subforests are listed in Algorithms~\ref{alg:LR} and \ref{alg:RL}.

\begin{algorithm2e}[htbp]
\caption{Construct an enumeration scheme for the subforests of a tree $T$ based on LR-postorder.\label{alg:LR}}
\SetKwInOut{Input}{input}
\SetKwInOut{Output}{output}

\Input{$T$, with $|T|=n$}
\Output{an enumeration sequence $L$ of subforests of $T$ based on the LR-postorder
}
\BlankLine
label the nodes of $T$ in LR-postorder \;
\For{$i \leftarrow 1$ \KwTo $n$}
{construct $S_i$ to be a sequence of subforests of $T[i]$ with the
rightmost root enumerated in
LR-postorder \;
}
$L = S_1$ \;
\For{$i \leftarrow 2$ \KwTo $n$}
{$L = L \circ S_i$ \;
}
output $L$ \;
\end{algorithm2e}

\begin{algorithm2e}[htbp]
\caption{Construct an enumeration scheme for the subforests of a tree $T$ based on RL-postorder.\label{alg:RL}}
\SetKwInOut{Input}{input}
\SetKwInOut{Output}{output}

\Input{$T$, with $|T|=n$}
\Output{an enumeration sequence $L$ of subforests of $T$ based on the RL-postorder
}
\BlankLine
label the nodes of $T$ in RL-postorder \;
\For{$i \leftarrow 1$ \KwTo $n$}
{construct $S_i$ to be a sequence of subforests of $T[i]$ with the
leftmost root enumerated in
RL-postorder \;
}
$L = S_1$ \;
\For{$i \leftarrow 2$ \KwTo $n$}
{$L = L \circ S_i$ \;
}
output $L$ \;
\end{algorithm2e}

A simple example of computing $d(T_1, T_2)$ is given in Figure~\ref{fig:t1t2-edit},
where the enumeration of nodes follows the LR-postorder
as described in Algorithm~\ref{alg:LR}.
A position in a table corresponding to a pair of nodes 
$(t_1[i], t_2[j])$ represents the distance between two
relevant subforests
with $t_1[i]$ and $t_2[j]$ being the rightmost roots.
Figures~\ref{fig:t1t2-edit_1} and \ref{fig:t1t2-edit_2} show
the computations for $d(T_1, T_2[j])$ with $t_2[j] \in \{d, e, f\}$
and $d(T_1[i], T_2)$ with $t_1[i] \in \{a, b\}$, respectively.
The computation for $d(T_1, T_2)$ is shown in Figure~\ref{fig:t1t2-edit_3} which makes use of the distances 
computed in Figures~\ref{fig:t1t2-edit_1} and \ref{fig:t1t2-edit_2}.
For example, consider the position corresponding to $(c, f)$
in Figure~\ref{fig:t1t2-edit_3}.
This corresponds to $d(T_1, T_2 - g)$ where $T_2 - g$ is the forest
obtained from $T_2$ by removing the root $g$.
By Equation~\ref{eqn:forest-dist}, we have:
\[ 
d(T_1, T_2 - g) = 
    \min\left\{
        \begin{array}{l}
          d(T_1 - c, T_2 - g) + \delta(c, \varnothing), \\
          d(T_1, T_2 - g - f) + \delta(\varnothing, f), \\
          d(\varnothing, T_2 - g - f) + d(T_1, f) 
        \end{array}
        \right\} \enspace .
\]
Denote by $D_i(x, y)$ with $i \in \{\subref{fig:t1t2-edit_1},
\subref{fig:t1t2-edit_2}, \subref{fig:t1t2-edit_3}\}$ the values in the tables
in Figures~\ref{fig:t1t2-edit_1}, \ref{fig:t1t2-edit_2}, and
\ref{fig:t1t2-edit_3}, respectively, at the position corresponding to
$x$ and $y$. 
Therefore, $d(T_1, T_2 - g) = 
\min\{D_{\subref{fig:t1t2-edit_3}}(b, f) + \delta(c, \varnothing),
D_{\subref{fig:t1t2-edit_3}}(c, e) + \delta(\varnothing, f),
D_{\subref{fig:t1t2-edit_3}}(\varnothing, e) + 
D_{\subref{fig:t1t2-edit_1}}(c, f)\} =
\min\{4+2, 4+2, 4+5\} = \min\{6, 6, 9\} = 6.$
Note that $d(T_1, f)$ in the last term is computed in the table of 
Figure~\ref{fig:t1t2-edit_1} at $D_{\subref{fig:t1t2-edit_1}}(c, f)$.
As another example, $D_{\subref{fig:t1t2-edit_3}}(c, d)$ and
$D_{\subref{fig:t1t2-edit_3}}(a, g)$ are computed at
$D_{\subref{fig:t1t2-edit_1}}(c, d)$ and 
$D_{\subref{fig:t1t2-edit_2}}(a, g)$, respectively.

\begin{figure}[htbp]
	\begin{center}
\renewcommand{\thesubfigure}{}\subfigure[$T_1$]{\includegraphics[scale=0.8]{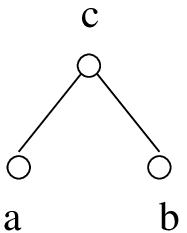}}
\qquad\qquad     
\renewcommand{\thesubfigure}{} \subfigure[$T_2$]{\includegraphics[scale=0.8]{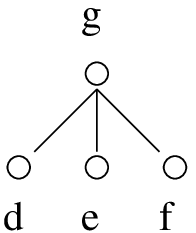}} \\ \vspace{5mm} 
\renewcommand{\thesubfigure}{(a)}
\subfigure[]{\label{fig:t1t2-edit_1}\scalebox{1}
{
	\begin{tabular}{|c|*{2}{c}c|}
			\hline
			 & $\varnothing$ & & $\{d, e, f\}$   \\
			\hline
			$\varnothing$ & 0 & & 2   \\ 
			              & $\uparrow$ & $\nwarrow$ &  \\
			a & 2 & & 1  \\
			              & $\uparrow$ & $\nwarrow$ & $\uparrow$ \\
			b & 4 & & 3  \\
			              &   & $\nwarrow$ & $\uparrow$ \\
			c & 6 & & 5  \\
			\hline
		\end{tabular}
}} 
\qquad\qquad     
\renewcommand{\thesubfigure}{(b)}
\subfigure[]{\label{fig:t1t2-edit_2}\scalebox{1}
{
	\begin{tabular}{|c|*{8}{c}c|}
			\hline
			 & $\varnothing$ & & d & & e & & f & & g  \\
			\hline
			$\varnothing$ & 0 & $\leftarrow$ & 2 & $\leftarrow$ & 4 & $\leftarrow$ & 6 & & 8  \\ 
			              &   & $\nwarrow$ &   & $\nwarrow$ &   & $\nwarrow$ &   & $\nwarrow$ & \\
			$\{a, b\}$ & 2 & & 1 & $\leftarrow$ & 3 & $\leftarrow$ & 5 & $\leftarrow$ & 7 \\
			\hline
		\end{tabular}
}} \\
\renewcommand{\thesubfigure}{(c)}
\subfigure[]{\label{fig:t1t2-edit_3}\scalebox{1}
{
	\begin{tabular}{|c|*{8}{c}c|}
			\hline
			 & $\varnothing$ & & d & & e & & f & & g  \\
			\hline
			$\varnothing$ & 0 & $\leftarrow$ & 2 & & 4 & & 6 & & 8  \\ 
			              &   & $\nwarrow$ &   & $\nwarrow$ &   & &   & & \\
			a & 2 & & 1 & $\leftarrow$ & 3 & & 5 & & 7 \\
			              &   & &   & $\nwarrow$ &   & $\nwarrow$ &   & & \\
			b & 4 & & 3 & & 2 & $\leftarrow$ & 4 & & 6 \\
			              &   & &   & &   & &   & $\nwarrow$ & \\
			c & 6 & & 5 & & 4 & & 6 & & 5 \\
			\hline
		\end{tabular}
}} 
  \end{center}
	\caption{Tables for the computation of $d(T_1, T_2)$. 
	The basic edit costs are defined as
	follows: $\delta(x, y)=1$ if $x \neq y$, and $0$ if $x=y$.
	$\delta(x, \varnothing)=\delta(\varnothing, x)=2$. The optimal edit scripts
	can be traced with the arrow sequences.}
	\label{fig:t1t2-edit}
\end{figure}

\begin{lemma} \label{lem:time-lr-enum}
The enumeration scheme based on the LR or RL-postorder takes
$O(|T|^2)$ steps.
\end{lemma}
\begin{proof}
We consider only the LR case as RL is symmetrical.
Each node $t_i$ within a subtree $T_k$ is contained in exactly one relevant subforest in $T_k$
having $t_i$ as the rightmost root. Denote by $s_i$ the number of
subtrees in which a node $t_i$ can be. 
Summing over all nodes, we have the total number of enumeration steps as $\sum_{i=1}^{|T|}s_i \leq \sum_{i=1}^{|T|}depth(t_i)
\leq \sum_{i=1}^{|T|}depth(T)
\leq \sum_{i=1}^{|T|}|T| = O(|T|^2)$. 
\end{proof}

\paragraph{Variable-Direction Recursion:}

For recursion of variable direction, we enumerate the subforests in one of two 
alternative orders as follows:

\begin{itemize}
\item Prefix-suffix postorder: For each node $t[i]$ enumerated in
LR-postorder, we enumerate the relevant subforests in increasing size as those
with distinct leftmost roots which contain $t[i]$ as the rightmost
root. 
\item Suffix-prefix postorder: For each node $t[i]$ enumerated in RL-postorder, we enumerate the relevant subforests in increasing size as those
with distinct rightmost roots which contain $t[i]$ as the leftmost
root. 
\end{itemize}

The order of enumeration would be such that for any subforest $F$,
all the subforests contained in $F$ have been enumerated
before $F$ is enumerated.
If we enumerate the subforests with the prefix-suffix postorder, this is done as follows. 
Consider in general a forest in which $t_i$ and $t_j$ are the leftmost and rightmost roots, respectively.
The rightmost root is enumerated in a left-to-right postorder
starting at the leftmost leaf.
For each $t_j$ thus enumerated, consider the largest forest
with $t_j$ being the rightmost root. Now, to obtain the order
for the subforests contained in this forest with $t_j$ being the
rightmost root, let 
$F_1, F_2, \cdots, F_k$ be the sequence of subforests resulted from
successively deleting the leftmost root from the forest until only
the rightmost subtree rooted on $t_j$ remains, i.e., $F_k = T[t_j]$. 
The order we want is the reverse sequence $F_k, F_{k-1}, \cdots, F_1$.
In this way, we obtain a sequence of subforests for each $t_j$.
Concatenate all the sequences in the increasing order of $t_j$, we
have the final sequence of all the subforests of $T$ arranged in a proper
order. 
The alternative way of enumerating the subforests, namely the 
suffix-prefix postorder, is handled
symmetrically. 
The procedures are listed in Algorithms~\ref{alg:prefix-suffix} and
\ref{alg:suffix-prefix}. 

\begin{algorithm2e}[htbp]
\caption{Construct an enumeration scheme for the subforests of a tree $T$ based on prefix-suffix postorder.\label{alg:prefix-suffix}}
\SetKwInOut{Input}{input}
\SetKwInOut{Output}{output}

\Input{$T$, with $|T|=n$}
\Output{an enumeration sequence $L$ of subforests of $T$ based on the prefix-suffix postorder
}
\BlankLine
construct $P$ to be a sequence of subforests of $T$ resulted from successive
deletion on the rightmost root \tcc*[r]{$P[1]=T$}
construct $P'=(F_1, F_2, \cdots, F_n)$ to be the reverse sequence of $P$ \tcc*[r]{$F_n=T$}
\For{$i \leftarrow 1$ \KwTo $n$}
{construct $S_i$ to be a sequence of subforests of $F_i \in P'$, all sharing the same rightmost root, resulted from successive
deletion on the leftmost root 
\tcc*[r]{$S_i[1]=F_i$, $S_i[k]=S_i[k-1]-lm\_root(S_i[k-1])$,
$rm\_root(S_i[k])=rm\_root(S_i[k-1])$, $\forall k > 1$}
construct $S'_i$ to be the reverse sequence of $S_i$ 
\tcc*[r]{$S'_i[|S'_i|]=F_i$}
}
$L = S'_1$ \;
\For(\tcc*[f]{concatenate all sequences}){$i \leftarrow 2$ \KwTo $n$}
{$L = L \circ S'_i$ \;
}
output $L$ \;
\end{algorithm2e}

\begin{algorithm2e}[htbp]
\caption{Construct an enumeration scheme for the subforests of a tree $T$ based on suffix-prefix postorder.\label{alg:suffix-prefix}}
\SetKwInOut{Input}{input}
\SetKwInOut{Output}{output}

\Input{$T$, with $|T|=n$}
\Output{an enumeration sequence $L$ of subforests of $T$ based on the suffix-prefix postorder
}
\BlankLine
construct $S$ to be a sequence of subforests of $T$ resulted from successive
deletion on the leftmost root \tcc*[r]{$S[1]=T$}
construct $S'=(F_1, F_2, \cdots, F_n)$ to be the reverse sequence of $S$ \tcc*[r]{$F_n=T$}
\For{$i \leftarrow 1$ \KwTo $n$}
{construct $P_i$ to be a sequence of subforests of $F_i \in S'$, all sharing the same leftmost root, resulted from successive
deletion on the rightmost root 
\tcc*[r]{$P_i[1]=F_i$, $P_i[k]=P_i[k-1]-rm\_root(P_i[k-1])$,
$lm\_root(P_i[k])=lm\_root(P_i[k-1])$, $\forall k > 1$}
construct $P'_i$ to be the reverse sequence of $P_i$ 
\tcc*[r]{$P'_i[|P'_i|]=F_i$}
}
$L = P'_1$ \;
\For(\tcc*[f]{concatenate all sequences}){$i \leftarrow 2$ \KwTo $n$}
{$L = L \circ P'_i$ \;
}
output $L$ \;
\end{algorithm2e}

Examples of prefix-suffix and suffix-prefix postorder enumerations  are given in Figure~\ref{fig:prefix-suffix} and
Figure~\ref{fig:suffix-prefix}, respectively.
In Figure~\ref{fig:prefix-suffix}, subforests having the same rightmost
root are in contiguous boxes, whereas in Figure~\ref{fig:suffix-prefix},
subforests having the same leftmost
root are in contiguous boxes.

\begin{figure}[htbp]
\begin{center}
\includegraphics[scale=0.6]{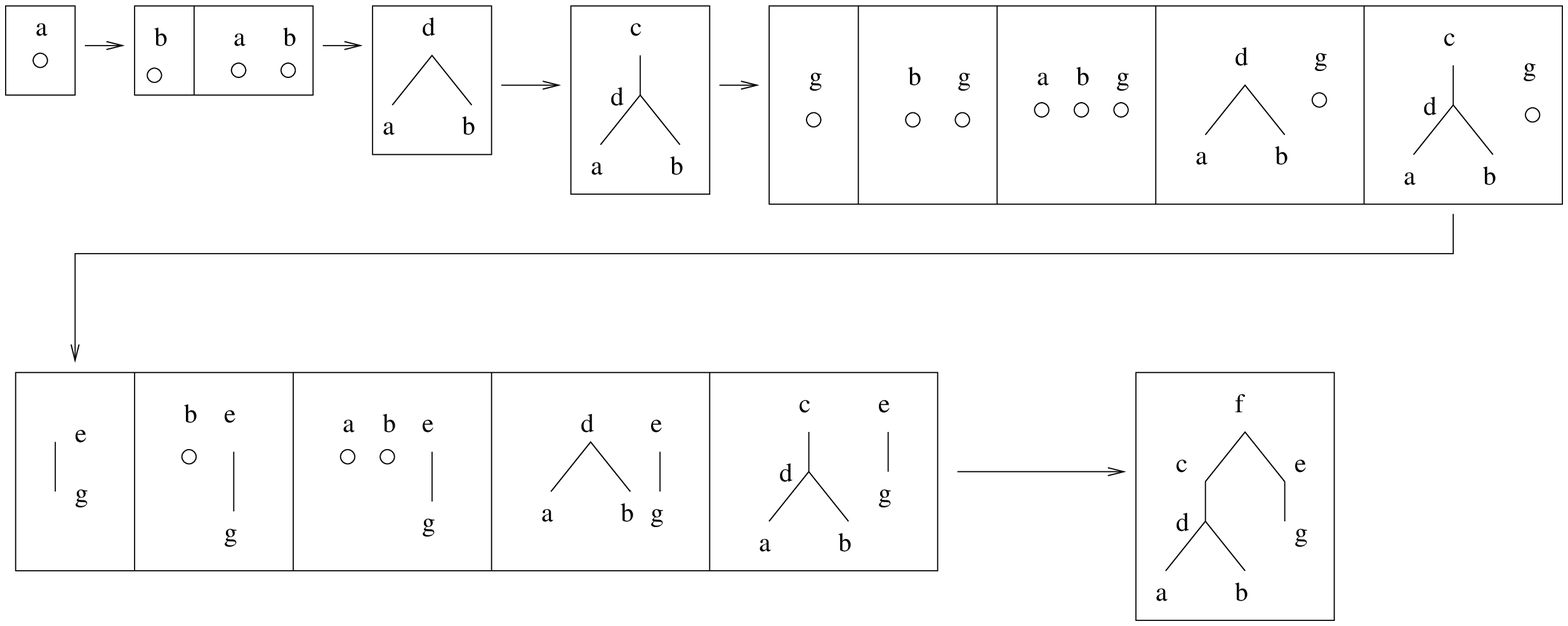}
\end{center}
\caption{An example of enumerating subforests in prefix-suffix
postorder.}
\label{fig:prefix-suffix}
\end{figure}

\begin{figure}[htbp]
\begin{center}
\includegraphics[scale=0.6]{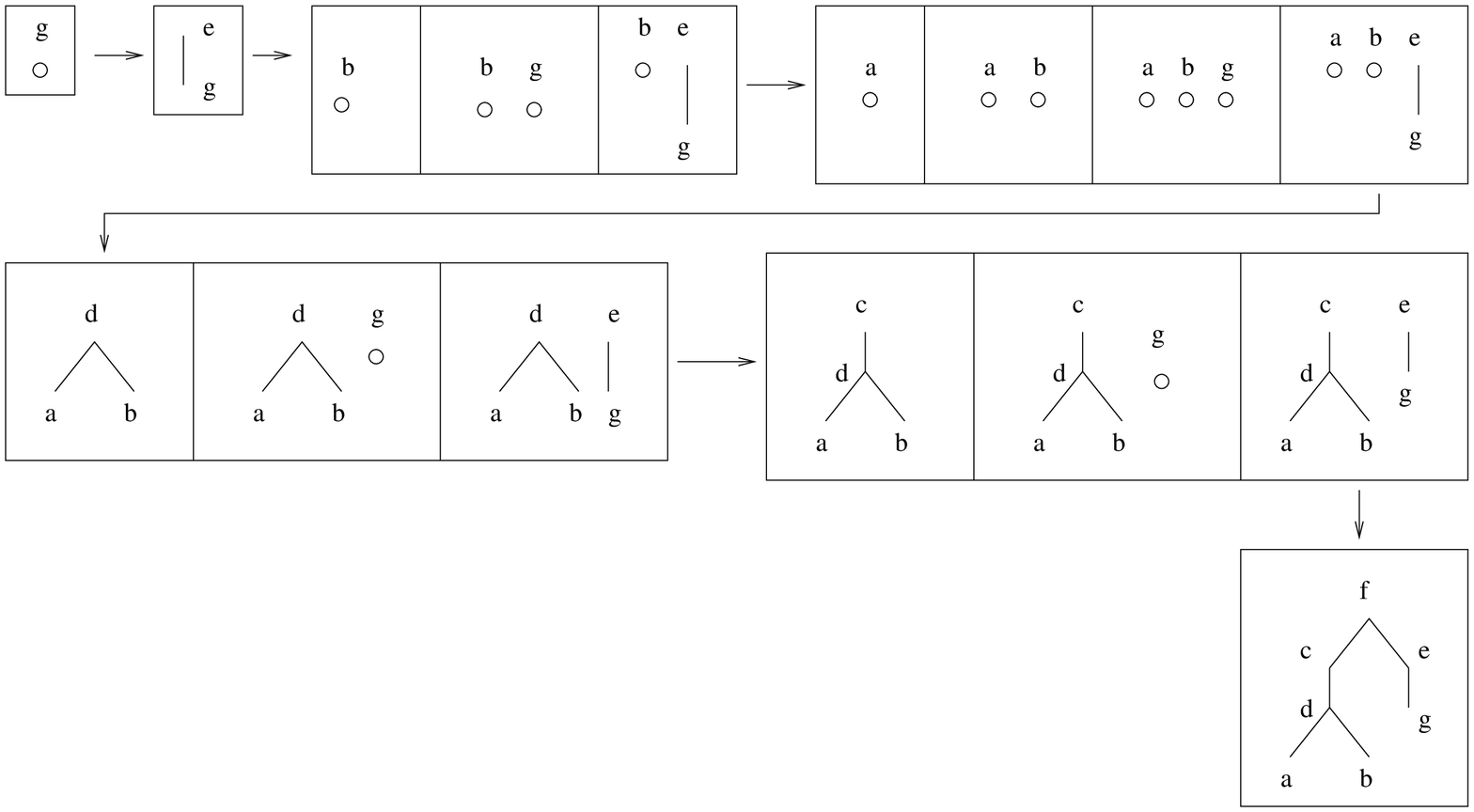}
\end{center}
\caption{An example of enumerating subforests in suffix-prefix
postorder.}
\label{fig:suffix-prefix}
\end{figure}

\begin{lemma} \label{lem:time-ps-enum}
The enumeration scheme based on prefix-suffix or suffix-prefix postorder
takes $O(|T|^2)$ steps.
\end{lemma}
\begin{proof}
We consider only the prefix-suffix postorder as the suffix-prefix postorder
is the symmetrical case. Denote by $f_i$ the number of subforests
with distinct leftmost roots which contain $t_i$ as the rightmost
root. 
Summing over
all nodes, we have 
$\sum_{i=1}^{|T|}f_i \leq \sum_{i=1}^{|T|}|T| = O(|T|^2)$.
\end{proof}

An algorithm for computing tree edit distances where the relevant
subforests are enumerated by the above procedures is given in
Algorithm~\ref{alg:TED_n4}. The algorithm can be implemented using 
$O(|T_1| \times |T_2|)$ space if the forest distances are allowed to be
overwritten. 

\begin{algorithm2e}[htbp]
\caption{Compute tree edit distance in $O(m^2n^2)$ time.\label{alg:TED_n4}}
\SetKwInOut{Input}{input}
\SetKwInOut{Output}{output}

\Input{$(T_1, T_2)$, with $|T_1|=m$ and $|T_2|=n$}
\Output{$d(T_1[i], T_2[j])$ for 
$1 \leq i \leq m$ and $1 \leq j \leq n$
}
\BlankLine
sort relevant subforests of $(T_1, T_2)$ into $(L_1, L_2)$ as
in Algorithms~\ref{alg:LR}, \ref{alg:RL}, \ref{alg:prefix-suffix},
or \ref{alg:suffix-prefix} \;
\For{$i \leftarrow 1$ \KwTo $|L_1|$}
{\For{$j \leftarrow 1$ \KwTo $|L_2|$}
{compute $d(L_1[i], L_2[j])$ as in Equation~\ref{eqn:forest-dist} \;
}}
\end{algorithm2e}

\begin{theorem}
The tree edit distance as computed in Algorithm~\ref{alg:TED_n4} takes
$O(m^2n^2)$ time, where $m=|T_1|$, and $n=|T_2|$. 
\end{theorem}
\begin{proof}
The result follows directly from Lemma~\ref{lem:time-lr-enum}
and \ref{lem:time-ps-enum}.
\end{proof}

The algorithms presented in this section follow a bottom-up 
dynamic programming style
where the tree nodes are numbered in postorder, in contrast to
the preorder numbering of nodes in Tai's algorithm~\cite{Tai}.
The way Tai's algorithm works is to progressively increase the
sizes of the trees, by one node at a time following the preorder numbers, and compute the distance for each such pair of
partial trees\footnote{In fact, it does not compute the true
distance since it only considers the optimal mappings along a pair
of paths for each pair of partial trees, instead of the entire 
partial trees. If Algorithm~\ref{alg:TED_n4} is applied for each
pair of partial trees, the time complexity is easily seen as
$O(m^3n^3)$.}. 

\section{Improved Algorithmic Strategies} \label{sec:strategies}

The algorithm presented in the previous section is based on
the principle of dynamic programming which relies on a
well-defined scheme for enumerating the relevant subforests. 
In this approach, forest distances are arranged in a certain order
so as to facilitate the relay of distance computations.
Essentially, we take advantage of the overlap among subforests
that are contained in the same subtree. 
To make further improvement, we look for ways to
take advantage of the overlap among subtrees as well. 

\subsection{Leftmost Paths} \label{sec:leftmost-paths}

We examine recursion of fixed direction, say rightmost 
recursion, the situation for leftmost recursion being
symmetrical. This means that the enumeration will be in LR-postorder.
Consider a path $(t_1, t_2, \cdots, t_k)$ where $t_i$ is the 
leftmost child of $t_{i+1}$ for $1 \leq i \leq k-1$. 
Let $(T_1, T_2, \cdots, T_k)$ be the sequence of subtrees
where $t_i$ is the root of $T_i$, and 
$(F_1, F_2, \cdots, F_k)$ be the sequence of sets where
$F_i$ denotes the set of subforests of $T_i$ all containing the
leftmost leaf of $T_i$.
We have $F_1 \subset F_2 \subset \cdots \subset F_k$. 
This means that enumerating $F_k$ once effectively takes care of
the enumerations for
$F_1, F_2, \cdots, F_{k-1}$. 
To generalize this situation to the whole tree, we see that
all subtrees sharing the same leftmost leaf can be handled together.
Carried out in this way, a tree is recursively decomposed into disjoint leftmost paths
where each such leftmost path is shared by a set of subtrees which can be handled together along this path with the LR-postorder
enumeration thereby removing the repetitions. 
This strategy was developed by 
Zhang and Shasha~\cite{Zhang:Shasha:TreeEditDistance:SIAMJC1989}.
An example of such path decomposition is given in Figure~\ref{fig:paths_L}. 

\begin{figure}[htbp]
  \begin{center}
    \subfigure[leftmost paths]{\label{fig:paths_L}\includegraphics[scale=0.34]{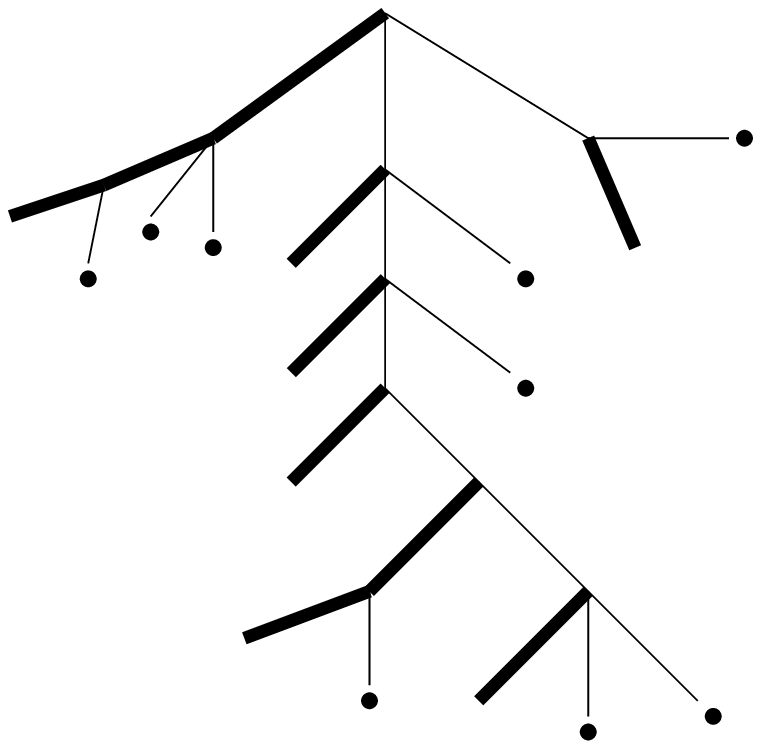}} 
\qquad\qquad \subfigure[rightmost paths]{\label{fig:paths_R}\includegraphics[scale=0.34]{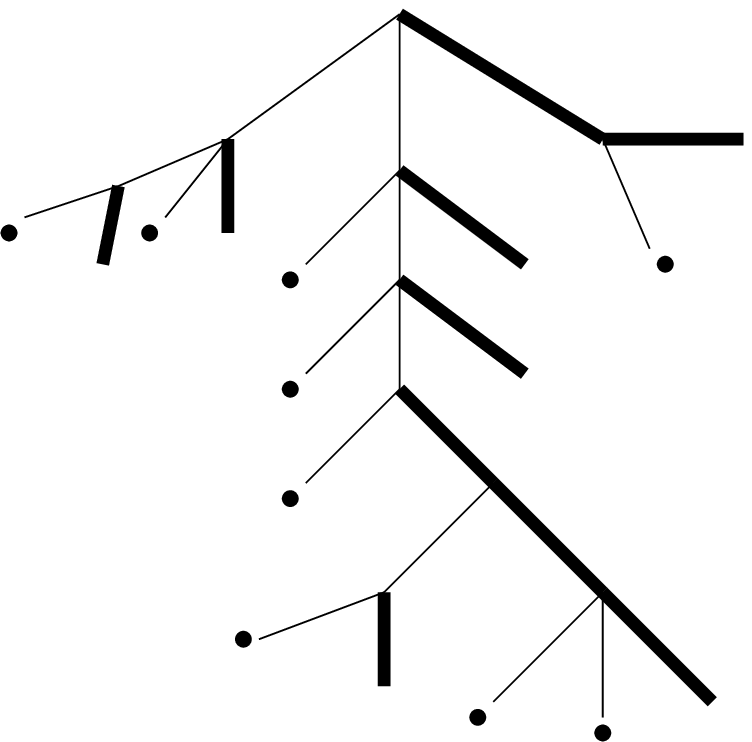}}
  \end{center}
  \caption{Leftmost paths and rightmost paths (in thick edges).}
  \label{fig:paths}
\end{figure}

Each leftmost path corresponds to the smallest subtree that
contains this path, and the root of this subtree is referred
to as an ``LR-keyroot'', which is defined as follows. 

\begin{definition}[LR-keyroots]
An LR-keyroot is either the root of $T$ or has a left sibling.
\end{definition}

The new enumeration scheme works as follows.
We identify all the LR-keyroots in the tree, and sort them
in increasing order by their LR-postorder numbers, referred to 
as ``LR-keyroot postorder''.
This will be the order by which the subforests are enumerated, i.e.,
based on the LR-keyroots with which they are associated. 
The procedure is listed in Algorithm~\ref{alg:LR-keyroot}. 

\begin{algorithm2e}[htbp]
\caption{Construct the enumeration scheme for the subforests of a tree $T$ in LR-keyroot postorder. The RL-based procedure is symmetrical to this. \label{alg:LR-keyroot}}
\SetKwInOut{Input}{input}
\SetKwInOut{Output}{output}

\Input{$T$, with $|T|=n$}
\Output{an enumeration sequence $L$ of subforests of $T$ in the LR-keyroot postorder
}
\BlankLine
identify the LR-keyroots of $T$ \;
sort the LR-keyroots in increasing order of LR-postorder numbers into a list $K=\left\{k_1, k_2, \cdots, k_l\right\}$ \;
\For{$i \leftarrow 1$ \KwTo $l$}
{construct $S_i$ to be a sequence of subforests of $T[k_i]$ 
with the rightmost root enumerated in
LR-postorder \;
}
$L = S_1$ \;
\For{$i \leftarrow 2$ \KwTo $n$}
{$L = L \circ S_i$ \;
}
output $L$ \;
\end{algorithm2e}

This enumeration scheme gives rise to the algorithm in 
Algorithm~\ref{alg:TED_zs}.

\begin{algorithm2e}[htbp]
\caption{Compute tree edit distance in $O(mn \prod_{i=1}^{2} \min\{depth(T_i),\#leaves(T_i)\})$ time.\label{alg:TED_zs}}
\SetKwInOut{Input}{input}
\SetKwInOut{Output}{output}

\Input{$(T_1, T_2)$, with $|T_1|=m$ and $|T_2|=n$}
\Output{$d(T_1[i], T_2[j])$ for 
$1 \leq i \leq m$ and $1 \leq j \leq n$
}
\BlankLine
sort relevant subforests of $(T_1, T_2)$ into $(L_1, L_2)$ as
in Algorithm~\ref{alg:LR-keyroot} \;
\For{$i \leftarrow 1$ \KwTo $|L_1|$}
{\For{$j \leftarrow 1$ \KwTo $|L_2|$}
{compute $d(L_1[i], L_2[j])$ as in Equation~\ref{eqn:forest-dist} \;
}}
\end{algorithm2e}

\begin{theorem} 
The algorithm computes $d(T_1[i], T_2[j])$ for all 
$1 \leq i \leq |T_1|$ and $1 \leq j \leq |T_2|$.
\end{theorem}
\begin{proof}
We prove it by induction on the sizes of the subtrees induced by the
keyroots.

\noindent Base case: This involves only the singleton subtrees. Since all the
basic edit costs with respect to single nodes are already defined, the base case holds.

\noindent Induction hypothesis: For any $(i, j) \in \{(i, j) \mid 
i \in LR\textnormal{-}keyroots(T_1),\  
j \in LR\textnormal{-}keyroots(T_2)\}$,
just before the computation of $d(T_1[i], T_2[j])$, 
the following set of distances have been computed, 
$D = D_1 \cup D_2$ where
\begin{itemize}
\item $D_1 = \{d(T_1[i'], T_2[j']) \mid 
i' \in T_1[i] - leftmost\textnormal{-}path(T_1[i]),\
j' \in T_2[j]\}$, 
\item $D_2 = \{d(T_1[i'], T_2[j']) \mid
i' \in T_1[i],\
j' \in T_2[j] - leftmost\textnormal{-}path(T_2[j])\}$.
\end{itemize}

\noindent Induction step: We show that 
$\{d(T_1[i'], T_2[j']) \mid i' \in T_1[i],\  
j' \in T_2[j]\}$ are all computed. The subtree-subtree distances
to be computed in the process of computing $d(T_1[i], T_2[j])$
are 
$\{d(T_1[i'], T_2[j']) \mid
i' \in leftmost\textnormal{-}path(T_1[i]),\
j' \in leftmost\textnormal{-}path(T_2[j])\}$.
The induction step holds since it is in accord
with the LR-keyroot postorder that the algorithm follows, 
which means that all distances specified
in the induction hypothesis have been computed.
This concludes the proof.
\end{proof}

To see the impact of the leftmost-path decomposition scheme on the time complexity, it is necessary to introduce
the concept of ``LR-collapsed depth'' defined as follows.

\begin{definition}[LR-Collapsed Depth] \label{def:lr-coll-depth}
The LR-collapsed depth of a node $t_i$ is the number of
its ancestors that are LR-keyroots.
The LR-collapsed depth of a tree $T$ is defined as
$LR\textnormal{-}collapsed\textnormal{-}depth(T) = 
\max\left\{LR\textnormal{-}collapsed\textnormal{-}depth(t_i) \mid t_i \in T\right\}$.
\end{definition}

Intuitively, the LR-collapsed depth of a tree $T$ represents the maximal number of non-leaf
LR-keyroots that a path in $T$ may contain. 
We define LR-collapsed depth as a way to estimate the maximal times a node, representing the rightmost root
of some relevant subforest, is enumerated with the LR-keyroot postorder. 
As a consequence of this enumeration scheme, repetitious enumerations
involving a given node are removed since
subtrees containing this node as well as having the same leftmost leaf
are no longer handled separately.

\begin{lemma} \label{lem:lr-coll-depth}
$LR\textnormal{-}collapsed\textnormal{-}depth(T) \leq \min\left\{depth(T), \#leaves(T)\right\}$.
\end{lemma}
\begin{proof}
Since the number of LR-keyroots on any path is bounded by the
depth of the path, we have 
$LR\textnormal{-}collapsed\textnormal{-}depth(T) \leq depth(T)$.
For any two LR-keyroots $k_i$ and $k_j$, the subtrees $T_i$ and $T_j$
rooted at $k_i$ and $k_j$ have distinct leftmost leaves. This means
that
the number of subtrees in $T$ that are rooted at LR-keyroots can not
exceed the number of leaves, i.e., 
$\#LR\textnormal{-}keyroots(T) \leq \#leaves(T)$. 
Since the number of LR-keyroots on any path is no more than the total
number of LR-keyroots in the tree, i.e.,
$LR\textnormal{-}collapsed\textnormal{-}depth(T) \leq \#LR\textnormal{-}keyroots(T)$,
we have $LR\textnormal{-}collapsed\textnormal{-}depth(T) \leq \#leaves(T)$.
Therefore, $LR\textnormal{-}collapsed\textnormal{-}depth(T)$ can be bounded by
$depth(T)$ or $\#leaves(T)$, whichever is smaller.
This concludes the proof.
\end{proof}

Here is the implication of Lemma~\ref{lem:lr-coll-depth}.
In the previous procedure, a node in $T$ may be enumerated 
$depth(T)$ times with the LR-postorder enumeration scheme,
because the maximal number of subtrees in which a node
may be contained is $depth(T)$. 
Grouping together subtrees with the same leftmost leaf
can remove the repetitions, and the improvement is evident since
the upper bound is reduced from $depth(T)$ to
$\min\left\{depth(T), \#leaves(T)\right\}$. 

\begin{theorem}
The tree edit distance problem can be solved in 
$O(mn \prod_{i=1}^{2} \min\{depth(T_i),\#leaves(T_i)\})$ time, where
$m = |T_1|$ and $n = |T_2|$.
\begin{proof}
From Lemma~\ref{lem:lr-coll-depth}, each node, representing the rightmost root of some relevant subforest, in $T$ is enumerated
at most $LR\textnormal{-}collapsed\textnormal{-}depth(T)$ times using
the enumeration scheme in Algorithm~\ref{alg:LR-keyroot}.
Hence, the result follows directly. 
\end{proof}
\end{theorem}

\begin{theorem}
The tree edit distance problem can be solved in 
$O(mn)$ space, where
$m = |T_1|$ and $n = |T_2|$.
\end{theorem}
\begin{proof}
The computation uses two $m \times n$ tables $D_t$ and $D_f$.
The forest-forest distances are computed in $D_f$ where the values
can be overwritten when the computation moves from one pair of
subtrees to another pair. The subtree-subtree distances obtained
in the process of computing the forest-forest distances are stored
in $D_t$, and fetched for use in computing forest-forest distances.
\end{proof}

In this section, a new way is presented for enumerating the
relevant subforests in LR-postorder where repetitious steps
associated with the leftmost paths in a tree
are eliminated, resulting in an improved time complexity. 
However, depending on the shapes of the trees, the leftmost-path decomposition for some tree shapes could yield
marginal benefits regarding the running time.  
This leads to the strategy to be presented in the next section. 

\subsection{Heavy Paths on One Tree} \label{sec:heavy-paths-n3logn}

We see from the previous section that the computation time
is due to the enumeration of subforests where each enumeration step
counts a constant time in performing a few simple arithmetics. 
The leftmost-path strategy improves the time complexity by enumerating subtrees with overlapping leftmost
paths together in the same sequence of computation. 
Since the running time is dependent on the shapes
of the trees,
it is worthwhile to consider a different type of path decomposition
that can also offer benefits with respect to the complexity. 
This possibility was explored and a new decomposition
strategy based on a type of path referred to 
as ``heavy path'' is due to Klein~\cite{klein:98}.
In contrast to the Zhang-Shasha strategy, which may be seen as
a way of improving upon the naive fixed-direction procedure based on the LR-postorder
enumeration scheme given in Section~\ref{sec:ted}, the new 
strategy may be seen as a way of improving upon the
variable-direction procedure based on the prefix-suffix or
suffix-prefix postorder enumeration scheme.
We give a few definitions related to the idea behind heavy path.

\begin{definition}[Heavy Child/Node] \label{def:heavy-child}
For any node $t$ in $T$, the child $t_h$ which is the root of the
largest subtree (breaking tie arbitrarily) among the sibling subtrees is the heavy child of $t$. We use the terms ``heavy child'' and
``heavy node'' interchangeably. 
\end{definition}

The definition of heavy path is given as follows.

\begin{definition}[Heavy Path]~\cite{Sleator:Tarjan:1983,Harel-Tarjan:SIAM1984} \label{def:heavy-path}
The heavy path of a tree $T$ is a unique path connecting the
root and a leaf of $T$ on which every node, except the root,
is a heavy node. 
\end{definition}

Figure~\ref{fig:paths_H} shows an example of a tree recursively decomposed into
a set of heavy paths.

\begin{figure}[htbp]
  \begin{center}
  \includegraphics[scale=0.34]{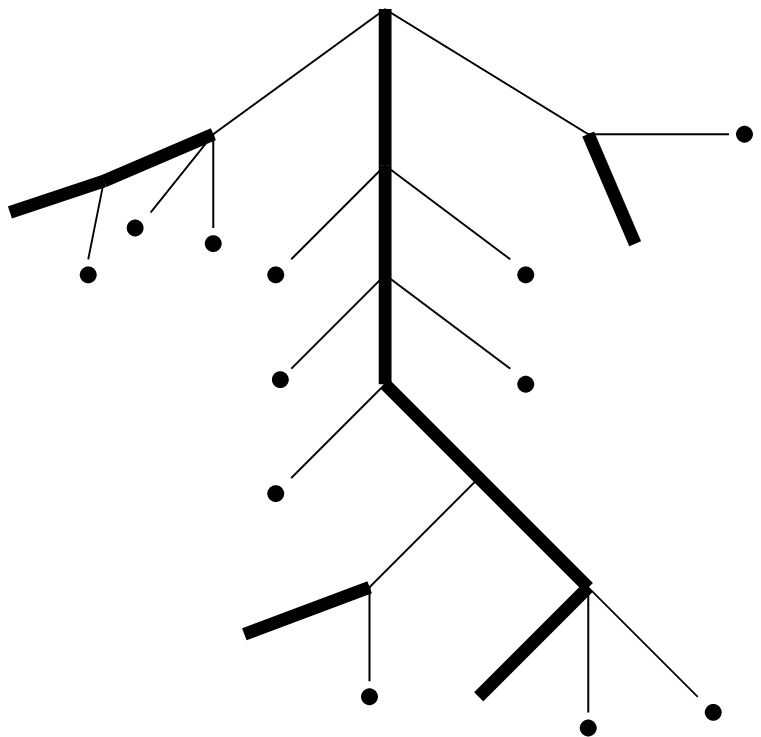}
  \end{center}
  \caption{Heavy paths (in thick edges).}
  \label{fig:paths_H}
\end{figure}

Similar to LR and RL-postorder which are defined with respect to
the leftmost path and rightmost path, respectively, we define an enumeration scheme
with respect to the heavy path as follows.

\begin{definition}[H-Postorder] \label{def:h-postorder}
The nodes in tree $T$ is enumerated in H-postorder as follows.
Start at the leaf $t_l$ on $heavy\textnormal{-}path(T)$, enumerate
the subtrees rooted on
its right siblings, if any, in LR postorder, 
then the subtrees rooted on its left siblings, if any, in RL postorder.
Continue and repeat the same process with each next higher node on
$heavy\textnormal{-}path(T)$ until reaching $root(T)$.
\end{definition}

If we ignore what happens on the left side of the heavy-path during
an H-postorder enumeration, then we see a sequence of enumeration
steps identical to an LR-postorder enumeration. 
If we ignore what happens on the right side of the heavy-path during
an H-postorder enumeration, then we see a sequence of enumeration
steps identical to an RL-postorder enumeration. 
Alternatively, a second version symmetrical to this one, i.e., RL then LR
intermittently, also works. In the following presentation, the 
version in Definition~\ref{def:h-postorder} is used.
An example of enumerating subforests in H-postorder is given in
Figure~\ref{fig:h-postorder}.
\begin{figure}[htbp]
\begin{center}
\includegraphics[scale=0.6]{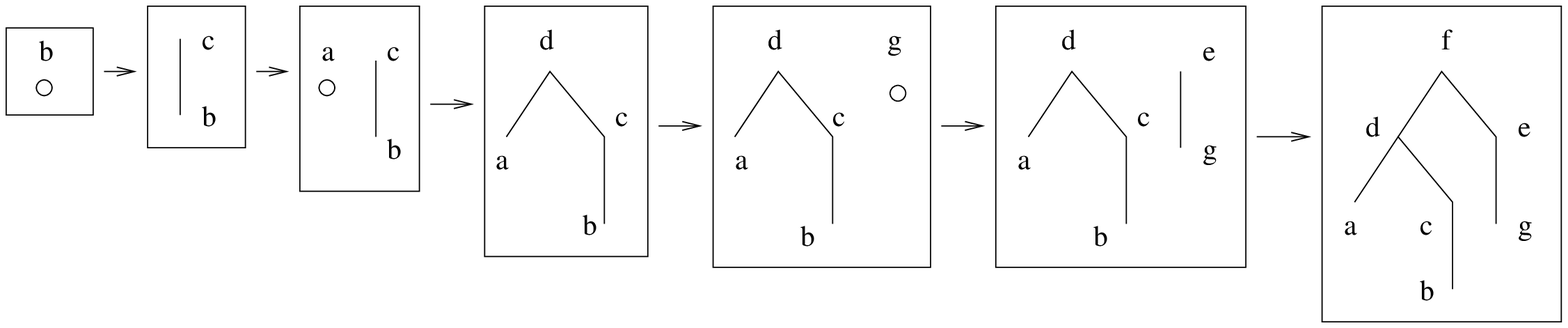}
\end{center}
\caption{An example of enumerating subforests in H-postorder.}
\label{fig:h-postorder}
\end{figure}

Analogous to LR-keyroots, a type of keyroots specific to this context
is defined as follows.
\begin{definition}[H-keyroots] \label{def:hkeyroots}
An H-keyroot is either the root of $T$ or the root of a subtree
in $T$ that has a larger sibling subtree. If multiple subtrees
are equally the largest among their sibling subtrees, all but one (chosen
arbitrarily) are H-keyroots.
\end{definition}

Definitions \ref{def:heavy-child} and \ref{def:hkeyroots} are
equivalent since for any node, once its heavy child is specified, the other children are H-keyroots, and vice versa.
A node in a tree is either a heavy node or an H-keyroot. 

The algorithm works as follows.
The H-keyroots in the larger tree are sorted into a list $L_1$
in increasing H-postorder numbers.
For each subtree of which the root is in $L_1$, order the
relevant subforests in H-postorder, and concatenate all
the ordered sequences to form the entire sequence as listed in Algorithm~\ref{alg:H-keyroot}, which we call the ``H-keyroot postorder''. 
On the smaller
tree, all subforests are ordered into a list $L_2$ in prefix-suffix or suffix-prefix
postorder, as in Algorithms~\ref{alg:prefix-suffix} or \ref{alg:suffix-prefix}.
The new algorithm is listed in Algorithm~\ref{alg:TED_klein}.

\begin{algorithm2e}[htbp]
\caption{Construct the enumeration scheme for the subforests of a tree $T$ in H-keyroot postorder. \label{alg:H-keyroot}}
\SetKwInOut{Input}{input}
\SetKwInOut{Output}{output}

\Input{$T$, with $|T|=n$}
\Output{an enumeration sequence $L$ of subforests of $T$ in the H-keyroot postorder
}
\BlankLine
identify the H-keyroots of $T$ \;
sort the H-keyroots in increasing order of H-postorder numbers into a list $K=\left\{k_1, k_2, \cdots, k_l\right\}$ \;
\For{$i \leftarrow 1$ \KwTo $l$}
{construct $S_i$ to be a sequence of subforests of $T[k_i]$ enumerated in
H-postorder \;
}
$L = S_1$ \;
\For{$i \leftarrow 2$ \KwTo $n$}
{$L = L \circ S_i$ \;
}
output $L$ \;
\end{algorithm2e}

\begin{algorithm2e}[htbp]
\caption{Compute tree edit distance in $O(m^2n\log n)$ time.\label{alg:TED_klein}}
\SetKwInOut{Input}{input}
\SetKwInOut{Output}{output}

\Input{$(T_1, T_2)$, with $|T_1|=m$, $|T_2|=n$, and $m \leq n$}
\Output{$d(T_1[i], T_2[j])$ for 
$1 \leq i \leq m$ and $1 \leq j \leq n$
}
\BlankLine
sort relevant subforests of $T_1$ into $L_1$ as
in Algorithms~\ref{alg:prefix-suffix} or \ref{alg:suffix-prefix},
and $T_2$ into $L_2$ as in Algorithm~\ref{alg:H-keyroot} \;
\For{$i \leftarrow 1$ \KwTo $|L_1|$}
{\For{$j \leftarrow 1$ \KwTo $|L_2|$}
{compute $d(L_1[i], L_2[j])$ as in Equation~\ref{eqn:forest-dist} \;
}}
\end{algorithm2e}

\begin{theorem} 
The algorithm computes $d(T_1[i], T_2[j])$ for all 
$1 \leq i \leq |T_1|$ and $1 \leq j \leq |T_2|$.
\end{theorem}
\begin{proof}
We prove it by induction on the sizes of the subtrees induced by the
keyroots.

\noindent Base case: This involves only the singleton subtrees. Since all the
basic edit costs with respect to single nodes are already defined, the base case holds.

\noindent Induction hypothesis: For any 
$k \in \{k \mid k \in H\textnormal{-}keyroots(T_2)\}$, 
just before the computation of $d(T_1, T_2[k])$, 
$\{d(T_1[i], T_2[j]) \mid
i \in T_1,\ j \in T_2[k] - heavy\textnormal{-}path(T_2[k])\}$
have been computed.

\noindent Induction step: We show that 
$\{d(T_1[i], T_2[j]) \mid i \in T_1,\ j \in T_2[k]\}$ are all computed. The subtree-subtree distances
to be computed in the process of computing $d(T_1, T_2[k])$
are $\{d(T_1[i], T_2[j]) \mid i \in T_1,\ j \in heavy\textnormal{-}path(T_2[k])\}$.
The induction step holds since it is in accord
with the postorder that the algorithm follows, which means that all distances specified
in the induction hypothesis have been computed.
This concludes the proof.
\end{proof}

We consider some aspects of the time complexity for this algorithm
as follows.

\begin{lemma} \label{lem:h-keyroots-sizes}
Let $h_1, h_2, \cdots, h_k$ be any sequence of H-keyroots that are on the same
path where $h_i$ is an ancestor of $h_j$ if $i < j$. Then, $|T[h_j]| \leq |T[h_i]|/2$ if $j = i+1$. 
\end{lemma}
\begin{proof}
Suppose that $|T[h_j]| > |T[h_i]|/2$. There are two cases to consider.

\begin{enumerate}
\item \label{case:light-parent} The nodes $h_i$ and $h_j$ are consecutive nodes on the path.
\item \label{case:heavy-parent} The nodes $h_i$ and $h_j$ are not consecutive nodes on the path.
\end{enumerate}

In case~\ref{case:light-parent}, $h_i$ is the parent of $h_j$.
If $|T[h_j]| > |T[h_i]|/2$, $h_j$ is the heavy child of $h_i$,
which is a contradiction to the fact that $h_j$ is an H-keyroot.
In case~\ref{case:heavy-parent}, there exists a node $t$ on the path
that is a descendent of $h_i$ as well as the parent of $h_j$.
Since $|T[h_j]| > |T[h_i]|/2$ and $|T[h_i]| > |T[t]|$, 
we have $|T[h_j]| > |T[t]|/2$.
This means that $h_j$ is the heavy child of $t$, contradicting the fact
that $h_j$ is an H-keyroot. This concludes the proof. 
\end{proof}

Analogous to LR-collapsed depth, a new version of collapsed depth
based on H-keyroots is defined as follows.

\begin{definition}[H-Collapsed Depth] \label{def:h-coll-depth}
The H-collapsed depth of a node $t_i$ is the number of
its ancestors that are H-keyroots.
The H-collapsed depth of a tree $T$ is defined as
$H\textnormal{-}collapsed\textnormal{-}depth(T) = 
\max\{H\textnormal{-}collapsed\textnormal{-}depth(t_i) \mid t_i \in T\}$.
\end{definition}

\begin{lemma} \label{lem:num-h-keyroots}
$H\textnormal{-}collapsed\textnormal{-}depth(T) \leq \log_2|T|$.
\end{lemma}
\begin{proof}
Consider a path $P$ in $T$ and the H-keyroots 
$h_0, h_1, h_2, \cdots, h_k$ on $P$ with $h_0$ being the root of $T$.
From Lemma~\ref{lem:h-keyroots-sizes},
each H-keyroot $h_i$ on $P$ is rooted at a subtree the size of which is
no larger than half the size of the subtree rooted at
$h_{i-1}$. Starting at $h_0$, 
traverse down the path $P$. 
For each subsequent H-keyroot that is being visited, the corresponding subtree size is reduced by at least a factor of 2 with respect to
the nearest H-keyroot previously visited. It takes at most $\log_2|T|$
encounters of H-keyroots for the subtree size to be reduced to 1,
which is also the maximal number of H-keyroots a node may have as
its ancestors. This concludes the proof.
\end{proof}

In contrast to LR-collapsed depth, H-collapsed depth has an
improved upper bound on the number of times that a node in the larger
tree may be enumerated, 
which is related to how many separate distance
computations, as identified by distinct keyroots, in which a node may participate. 
The bound, on the other hand, for
a node in the smaller tree to be enumerated is the size of the tree,
since all the subforests are considered. 
The overall impact on the time complexity is given in the next theorem.

\begin{theorem}
The tree edit distance problem can be solved in $O(m^2n\log n)$ time where $|T_1|=m$, $|T_2|=n$, and $m \leq n$.
\end{theorem}
\begin{proof}
For any $(i, j)$ with $i \in T_1$ and $j \in T_2$, $i$ is enumerated
the number of times equal the number of subforests with distinct
leftmost roots which contain $i$ as the rightmost root, or alternatively, the number of subforests with distinct
rightmost roots which contain $i$ as the leftmost root.
This is bounded by the size of $T_1$, i.e., $m$. On the other hand,
$j$ is enumerated at most $1+\log_2 n$ times according to
Lemma~\ref{lem:num-h-keyroots}, since this is the
upper bound on the number of subtrees in $T_2$ rooted on distinct
H-keyroots which contain $j$,
and in each one $j$ is enumerated once.
The result thus follows. 
\end{proof}

\begin{theorem} \label{thm:klein-space}
The new algorithm solves the tree edit distance problem in $O(mn)$ space where $|T_1|=m$, $|T_2|=n$, and $m \leq n$.
\end{theorem}
\begin{proof}
We use a $2 \times m^2$ table where the $m^2$ subforests in $T_1$
are arranged in prefix-suffix or suffix-prefix order.
For $T_2$, the idea is essentially a linear-space algorithm by which
distances for only one subforest are computed and updated when
moving to the next subforest in the enumeration sequence.
The subtree-subtree distances are stored in an $m \times n$ table. 
\end{proof}

In the next section, we see how this algorithm is improved by
a strategy that finds a way to apply heavy-path decompositions on both trees.

\subsection{Heavy Paths on Both Trees} \label{sec:heavy-paths-n3}

The algorithm by Klein reduces the upper bound on the number of
separate distance computations required 
from $O(\min\{depth(T),\#leaves(T)\})$ to $O(\log |T|)$
for one tree.
This is done at the cost of having to consider all the subforests
in the other tree.
Demaine \textit{et al.}~\cite{DMRW2009} improved this strategy
by a way that applies decompositions on both trees.
By their algorithm, $d(T_1, T_2)$ is computed as follows, assuming that
$|T_1| \leq |T_2|$:

\begin{enumerate}
\item \label{step:n3-d-T1-T2-swap} If $|T_1| > |T_2|$, compute $d(T_2, T_1)$.
\item \label{step:n3-d-T1-T2-rec} Recursively, compute $d(T_1, T_2[k])$ with $k$ being
the set of nodes connecting directly to $heavy\textnormal{-}path(T_2)$
with single edges.
\item \label{step:n3-d-T1-T2-bu} Compute $d(T_1, T_2)$ 
by enumerating relevant subforests of $T_1$ in prefix-suffix
(Algorithm~\ref{alg:prefix-suffix})
or suffix-prefix postorder (Algorithm~\ref{alg:suffix-prefix}), 
and relevant subforests of $T_2$ in H-postorder
(Definition~\ref{def:h-postorder}). 
\end{enumerate}

This is a combined recursive and bottom-up procedure where
the order of subtree-subtree pairs is arranged recursively
in step~\ref{step:n3-d-T1-T2-rec}, whereas 
the forest-forest distances encountered in a subtree-subtree
distance computation, in step~\ref{step:n3-d-T1-T2-bu}, are computed
with bottom-up enumerations. 
In comparison, the algorithm by Klein consists of only 
steps~\ref{step:n3-d-T1-T2-rec} and \ref{step:n3-d-T1-T2-bu},
without step~\ref{step:n3-d-T1-T2-swap}.
Due to step~\ref{step:n3-d-T1-T2-swap}, decomposition is done on
both trees. 
Here, step~\ref{step:n3-d-T1-T2-bu} differs from the procedure
in~\cite{DMRW2009} where the computation is done with recursion.
Nonetheless, they are equivalent since the precondition, that
the subtree-subtree distances related to step~\ref{step:n3-d-T1-T2-rec} have been obtained, is the same. 
These distances are: $d(T_1[i], T_2[j])$ for all $i \in T_1$ and 
$j \in T_2 - heavy\textnormal{-}path(T_2)$.
The subtree-subtree distances obtained in step~\ref{step:n3-d-T1-T2-bu} alone are $d(T_1[i], T_2[j])$ for all $i \in T_1$ and 
$j \in heavy\textnormal{-}path(T_2)$.
Therefore, the postcondition of step~\ref{step:n3-d-T1-T2-bu} is that
$d(T_1[i], T_2[j])$ for all $i \in T_1$ and $j \in T_2$ have all
been obtained.
To adapt the procedure into a bottom-up dynamic programming algorithm,
the order of computation sequence can be obtained in advance by
running the recursion of step~\ref{step:n3-d-T1-T2-rec},
and only recording the subtree pair in step~\ref{step:n3-d-T1-T2-bu}
without actually computing the distance.
This yields the bottom-up computation sequence. 

We now consider some aspects of the algorithm.

\begin{theorem} 
The algorithm computes $d(T_1[i], T_2[j])$ for all 
$1 \leq i \leq |T_1|$ and $1 \leq j \leq |T_2|$.
\end{theorem}
\begin{proof}
We prove it by induction on the sizes of the subtrees induced by the
keyroots.

\noindent Base case: This involves only the singleton subtrees. Since all the
basic edit costs with respect to single nodes are already defined, the base case holds.

\noindent Induction hypothesis: For any 
$(i, j) \in \{(i, j) \mid 
i \in H\textnormal{-}keyroots(T_1),\
j \in H\textnormal{-}keyroots(T_2)\}$,
after step~\ref{step:n3-d-T1-T2-rec},
\begin{enumerate}
\item if $|T_1[i]| \leq |T_2[j]|$, then 
$\{d(T_1[i'], T_2[j']) \mid
i' \in T_1[i],\ j' \in T_2[j] - heavy\textnormal{-}path(T_2[j])\}$
have been computed,
\item if $|T_1[i]| > |T_2[j]|$, then 
$\{d(T_1[i'], T_2[j']) \mid
i' \in T_1[i] - heavy\textnormal{-}path(T_1[i]),\
j' \in T_2[j]\}$ have been computed.
\end{enumerate}

\noindent Induction step: We show that 
$\{d(T_1[i'], T_2[j']) \mid i' \in T_1[i],\ j' \in T_2[j]\}$
are all computed.
The subtree-subtree distances
to be computed in step~\ref{step:n3-d-T1-T2-bu} are:

\begin{enumerate}
\item $\{d(T_1[i'], T_2[j']) \mid
i' \in T_1[i],\ j' \in heavy\textnormal{-}path(T_2[j])\}$,
if $|T_1[i]| \leq |T_2[j]|$, 
\item $\{d(T_1[i'], T_2[j']) \mid
i' \in heavy\textnormal{-}path(T_1[i]),\ j' \in T_2[j]\}$,
if $|T_1[i]| > |T_2[j]|$. 
\end{enumerate}

The induction step holds since it is in accord
with the postorder that the algorithm follows, which means that all distances specified
in the induction hypothesis have been computed.
This concludes the proof.
\end{proof}

Given two subtree pairs
$(T_1[k], T_2[k'])$ and $(T_1[l], T_2[l'])$ in which
$(i, j)$ is contained, with
$k, l \in H\textnormal{-}keyroots(T_1)$,
$k', l' \in H\textnormal{-}keyroots(T_2)$, 
$k = \min\{x \mid x \in ancestors(l) \cap H\textnormal{-}keyroots(T_1)\}$, and 
$k' = \min\{x \mid x \in ancestors(l') \cap H\textnormal{-}keyroots(T_2)\}$, 
we consider the possibilities pertaining to the relative sizes
of the four trees, where we use H-keyroots to represent the
corresponding subtree sizes.
We write $k \prec l$ if $|T[k]| > |T[l]|$.

\begin{enumerate}
\item $k \prec l \prec k' \prec l'$, as in Figure~\ref{fig:kl_1}, \label{case:kl-1}
\item $k \prec k' \prec l \prec l'$, as in Figure~\ref{fig:kl_2}, \label{case:kl-2}
\item $k \prec k' \prec l' \prec l$, as in Figure~\ref{fig:kl_3}. \label{case:kl-3}
\end{enumerate} 

\begin{figure}[htbp]
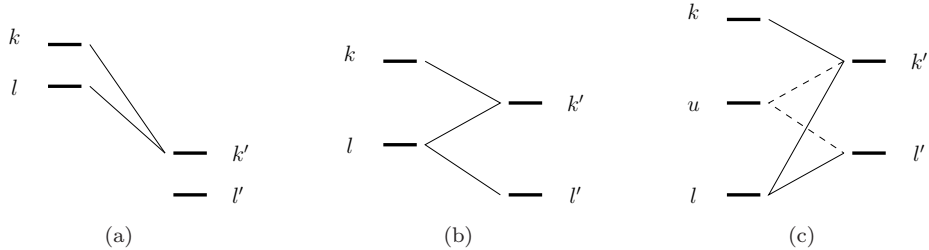

  \begin{center} 
  \subfigure[]{\label{fig:kl_1}\scalebox{0.7}{\input{kl-1.pstex_t}}} \qquad\qquad 
\subfigure[]{\label{fig:kl_2}\scalebox{0.7}{\input{kl-2.pstex_t}}} \qquad\qquad \subfigure[]{\label{fig:kl_3}\scalebox{0.7}{\input{kl-3.pstex_t}}}
  \end{center}
  \caption{Depiction of possible cases for relative subtree sizes due to
  heavy-path decomposition. A line between two size levels (thick lines) indicates that a distance computation is performed for subtrees
  of corresponding size levels. Here, $k, l \in H\textnormal{-}keyroots(T_1)$, and
$k', l' \in H\textnormal{-}keyroots(T_2)$. 
$T_1[k]$ and $T_2[k']$ decompose once to yield $T_1[l]$ and $T_2[l']$, 
respectively.
} \label{fig:kl}
\end{figure}

The distance computations in which a pair of nodes $(i, j)$ would participate
are represented by solid lines drawn between the size levels as shown
in Figure~\ref{fig:kl}. These situations arise as a result of only the
larger subtree being allowed to decompose. If we count the number of
enumeration steps involving $(i, j)$, the analysis is as follows.
We enumerate each node in the larger tree once, while
enumerate each node in the smaller subtree
a number of times no more than the size of the subtree.
This way of counting regarding the smaller subtree is based on 
how many subforests with distinct leftmost roots may include the node
as the rightmost root, or symmetrically, 
how many subforests with distinct rightmost roots may include the node
as the leftmost root.
For case~\ref{case:kl-1}, i.e., $k \prec l \prec k' \prec l'$, $i$ is counted
once for $k$ and $l$ each, while $j$ is counted $|T_2[k']|$ steps for $k'$, for a total of
$2|T_2[k']|$ steps.
For case~\ref{case:kl-2}, i.e., $k \prec k' \prec l \prec l'$, $(i, j)$ is counted
$1 \times |T_2[k']|$ steps for $(k, k')$, $|T_1[l]| \times 1$ steps for $(l, k')$, and
$1 \times |T_2[l']|$ steps for $(l, l')$, for a total of $|T_2[k']|+|T_1[l]|+|T_2[l']|$
steps. 
For case~\ref{case:kl-3}, i.e., $k \prec k' \prec l' \prec l$,
$(i, j)$ is counted $1 \times |T_2[k']|$ steps for $(k, k')$,
$|T_1[l]| \times 1$ steps for $(l, k')$, and $|T_1[l]| \times 1$ steps for
$(l, l')$, for a total of $|T_2[k']|+2|T_1[l]|$ steps. 
In the time complexity analysis, the steps in case~\ref{case:kl-3} can be
bounded by replacing $l$ by $u$ where $k' \prec u \prec l'$, which
results in the two lines incident to $l$ being replaced by the
two lines incident to $u$, returning back to case~\ref{case:kl-2}. This means that in the time complexity analysis, we only need to consider steps from 
case~\ref{case:kl-1} and case~\ref{case:kl-2}, as well as their
symmetrical counterparts. 
Figure~\ref{fig:n3-size-mapping} illustrates a situation
where $(i, j)$ are enumerated as a pair
in the worst case (i.e., $1+\log_2 m$ and $1+\log_2 n$ levels, respectively)
with respect to the sizes of the subtrees in which
$(i, j)$ are contained.

\begin{figure}[htbp]
  \begin{center}
  \scalebox{1}{\input{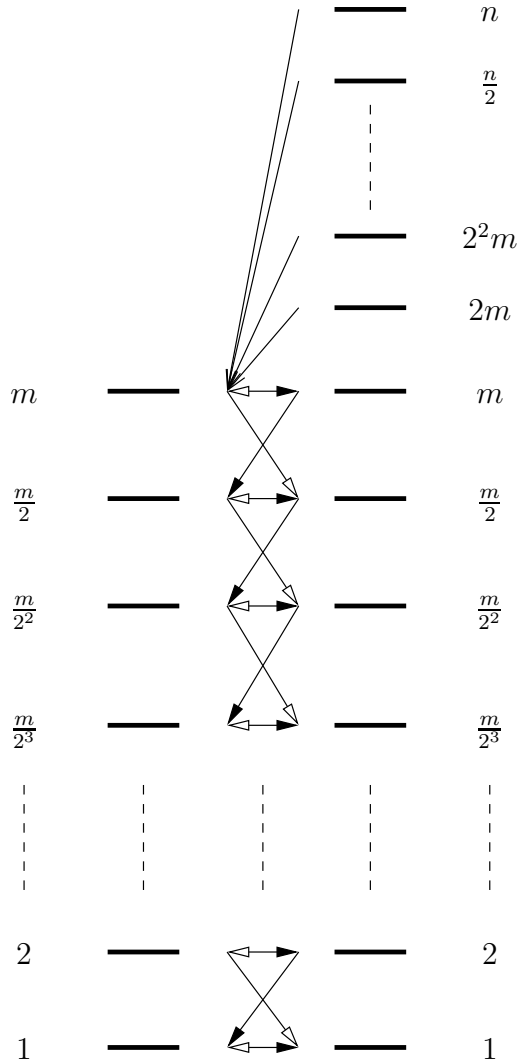}}
  \end{center}
  \caption{Depiction of the situation where $(i, j)$ are enumerated as a pair
  in the worst case (i.e., $1+\log_2 m$ and $1+\log_2 n$ levels, respectively)
  with respect to the sizes of subtrees in which
  $(i, j)$ may be contained. Levels of different sizes are represented
  by thick lines. A line is drawn between two size levels to indicate
  inclusion of $(i, j)$ where an arrowhead points to the smaller size.
  For size levels no more than $m$,
  two types of arrowheads (filled and hollow) are used to distinguish
  between alternative sequences of decompositions where the same sequence can
  be traced by following the lines with same type of arrowheads.}
  \label{fig:n3-size-mapping}
\end{figure}

The following lemma is based on an observation that is crucial
in obtaining the claimed time bound. 

\begin{lemma} \label{lem:sum-sizes}
Let $W = \{w_1, w_2, \cdots, w_k\}$ be a list of numbers satisfying
that for any $w_i, w_j \in W$, $w_j \leq \frac{w_i}{2}$ if $j = i+1$. Then, for any $u$,
$S = \sum_{i=u}^{k}w_i \leq 2w_u$.
\end{lemma}
\begin{proof}
Recall that $R = \sum_{i=0}^{n}\frac{1}{2^i} \leq 2$ for any $n \geq 0$,
which is proved by showing that $2R-R = 2 - \frac{1}{2^n} \leq 2$.
Therefore,
we have
$S = \sum_{i=u}^{k}w_i \leq w_{u}\sum_{i=0}^{k-u}\frac{1}{2^i} \leq 
2w_u$.
\end{proof}

The following theorem gives the result for the time complexity
of the algorithm. 

\begin{theorem}
The tree edit distance problem can be solved in $O(m^2n(1+\log\frac{n}{m}))$ time where $|T_1|=m$, $|T_2|=n$, and $m \leq n$.
\end{theorem}
\begin{proof}
For any $(i, j)$ where $i \in T_1$ and $j \in T_2$, we count
the number of times that $(i, j)$ is enumerated in distance
computations in all possible combinations based on the relative
sizes of the subtrees in which $i$ and $j$ are contained.
These combinations can be divided into three categories:

\begin{enumerate} 
\item \label{case-log} $(i, j) \in (T_1[h], T_2[h'])$
for some $(h, h')$, with $|T_1[h]| \leq m$ and $m < |T_2[h']| \leq n$,
\item \label{case-large-small} $(i, j) \in (T_1[h], T_2[h'])$
for some $(h, h')$, with $m \geq |T_1[h]| \geq |T_2[h']|$,
\item \label{case-small-large} $(i, j) \in (T_1[h], T_2[h'])$
for some $(h, h')$, with $|T_1[h]| \leq |T_2[h']| \leq m$.
\end{enumerate}

In the above cases, for each pair of nodes $(i, j)$ that participate
in a distance computation for a pair of subtrees,
the node in the larger subtree is counted once, while
the node in the smaller subtree is counted
a number of times no more than the size of the subtree.
This way of counting with respect to the smaller subtree is based on 
how many subforests with distinct leftmost roots may include the node
as the rightmost root, or symmetrically, 
how many subforests with distinct rightmost roots may include the node
as the leftmost root.

Let $S_1$, $S_2$, and $S_3$ be maximal numbers of total enumeration steps
corresponding to category~\ref{case-log}, \ref{case-large-small},
and \ref{case-small-large}, respectively. 

From Lemma~\ref{lem:num-h-keyroots}, a node in $T$, 
with $|T|=n \geq m$, can be in at most $1+\log_2 m$ subtrees of
sizes no more than $m$, rooted at distinct H-keyroots. 
Therefore,
$S_1 \leq m^2n(\log_2 n - \log_2 m) = m^2n\log_2\frac{n}{m}$.

For $S_2$ and $S_3$, we give a simplified analysis which
includes all combinations of which some are redundant
due to the fact that a smaller subtree does not decompose
until it becomes the larger one. 
This, however, does not change the complexity as the difference is
within a negligible factor, due to Lemma~\ref{lem:sum-sizes}. 
From Lemma~\ref{lem:h-keyroots-sizes}, 
\ref{lem:num-h-keyroots}, and \ref{lem:sum-sizes},
we have 
$S_2 \leq m(n \times 2m\sum_{i=0}^{\log_2m}\frac{1}{2^i}) \leq 4m^2n$, and
$S_3 \leq (m \times 2m\sum_{i=0}^{\log_2m}\frac{1}{2^i})n \leq 4m^2n$.
This yields a total number of steps in the worst case as
$S_1 + S_2 + S_3 = O(m^2n(1+\log\frac{n}{m}))$.

For a more accurate
estimate of $S_2$ and $S_3$ (see Figure~\ref{fig:n3-size-mapping}),
we have 
$S_2 \leq m \times nm + 2m \times (n\sum_{i=1}^{\log_2m}\frac{m}{2^i}) \leq 3m^2n$, and
$S_3 \leq mm \times n + (m\sum_{i=1}^{\log_2m}\frac{m}{2^i})\times 2n \leq 3m^2n$.
Hence, the total time is
$S_1 + S_2 + S_3 = O(m^2n(1+\log\frac{n}{m}))$.
\end{proof}

\paragraph{Remark:} It has been shown that
there exist trees for which 
$\Omega(m^2n(1+\log\frac{n}{m}))$ time is required to compute
the distance no matter what strategy is used~\cite{DMRW2009}.

\begin{theorem} 
The new algorithm solves the tree edit distance problem in $O(mn)$ space where $|T_1|=m$, $|T_2|=n$, and $m \leq n$.
\end{theorem}
\begin{proof}
The method is same as that described in Theorem~\ref{thm:klein-space}. 
\end{proof}

The efficiency of the algorithm can be tightened up by combining all three path
decomposition strategies (i.e., leftmost, rightmost, and heavy paths)
to yield an algorithm with the least total enumeration steps.
The basic idea, while retaining the general framework of the algorithm, is to recursively count the number of enumeration steps
resulted from different types of path decompositions
without actually carrying out the distance computations within the
original algorithm. This means that for any $d(T_1[i], T_2[j])$
to be computed by the algorithm, step~\ref{step:n3-d-T1-T2-bu} counts
the number of enumeration steps involving the nodes on the path for 
$(T_1[i], T_2[j])$ with respect to
each type of decomposition, while the steps involving other nodes
that do not belong to
the path are counted recursively at step~\ref{step:n3-d-T1-T2-rec}
(i.e., one recursive call for each type of path)
and combined with the counts in step~\ref{step:n3-d-T1-T2-bu} so
as to decide which path to use for that level.
The results from each level are recorded into a table which
consumes $O(mn)$ space. The recorded information is then
used to guide the distance computations in the selection of
strategy at each step. 
This yields an overall least total number of enumeration steps
with respect to all strategies considered.
The time complexity, however, remains the same due to the above
remark regarding the lower bound. 

\section{Conclusions} \label{sec:conclusion}
This article considers the tree edit distance problem and 
formulation of solutions in the form of recursion. 
In particular, 
a class of algorithms based on closely related decomposition
schemes for computing the tree edit distance between two
ordered trees are reviewed, with an attention to aspects of
time complexity analysis. 

As a summary of the contents presented in 
Section~\ref{sec:strategies}, we recapture the related path-decomposition
strategies as follows.

\paragraph{\textnormal{Leftmost paths:}}
$d(T_1, T_2)$ is computed as follows.
\begin{enumerate}
\item Recursively, compute $d(T_1[k], T_2)$ and $d(T_1, T_2[k'])$, with 
$k$ being
the set of nodes connecting directly to $leftmost\textnormal{-}path(T_1)$
with single edges, whereas
$k'$ being
the set of nodes connecting directly to $leftmost\textnormal{-}path(T_2)$
with single edges.
\item Compute $d(T_1, T_2)$ 
by enumerating relevant subforests of $T_1$ and $T_2$ in LR-postorder. 
\end{enumerate}

\paragraph{\textnormal{Heavy paths on one tree:}}
$d(T_1, T_2)$ is computed as follows.
\begin{enumerate}
\item Recursively, compute $d(T_1, T_2[k])$ with $k$ being
the set of nodes connecting directly to $heavy\textnormal{-}path(T_2)$
with single edges.
\item Compute $d(T_1, T_2)$ 
by enumerating relevant subforests of $T_1$ in prefix-suffix
or suffix-prefix postorder, 
and relevant subforests of $T_2$ in H-postorder. 
\end{enumerate}

\paragraph{\textnormal{Heavy paths on both trees:}}
$d(T_1, T_2)$ is computed
as follows, assuming that $|T_1| \leq |T_2|$.
\begin{enumerate}
\item If $|T_1| > |T_2|$, compute $d(T_2, T_1)$.
\item Recursively, compute $d(T_1, T_2[k])$ with $k$ being
the set of nodes connecting directly to $heavy\textnormal{-}path(T_2)$
with single edges.
\item Compute $d(T_1, T_2)$ 
by enumerating relevant subforests of $T_1$ in prefix-suffix
or suffix-prefix postorder, 
and relevant subforests of $T_2$ in H-postorder. 
\end{enumerate}

All of the above strategies can be equivalently
stated as applying
Equation~\ref{eqn:forest-dist} according to predefined directions without recursing into subproblems
already computed. 

\bibliographystyle{plain}
\bibliography{allbibs}
\end{document}